\documentclass[%
 reprint,
 amsmath,amssymb,
 aps,
]{revtex4-1}

\usepackage{graphicx}
\usepackage{dcolumn}
\usepackage{bm}

\usepackage{amsmath,amsfonts,amssymb,amscd,amsthm}
\usepackage{multirow}
\usepackage{color}

\newtheorem{lemma}{Lemma}

\newtheorem{theorem}[lemma]{Theorem}

\newtheorem{definition}[lemma]{Definition}

\begin{document}


\title{Strong Quantum Nonlocality without Entanglement in Multipartite Quantum Systems}

\author{Pei Yuan}
\email{yuanpei@ict.ac.cn}
\author{Guojing Tian}%
\email{tianguojing@ict.ac.cn}
\author{Xiaoming Sun}
 \email{sunxiaoming@ict.ac.cn}
\affiliation{%
 CAS Key Lab of Network Data Science and Technology, Institute of Computing Technology, Chinese Academy of Sciences, 100190, Beijing, China.\\
University of Chinese Academy of Sciences, Beijing, 100049, China
}%


\date{\today}

\begin{abstract}
In this paper, we generalize the concept of strong quantum nonlocality from two aspects. Firstly in $\mathbb{C}^d\otimes\mathbb{C}^d\otimes\mathbb{C}^d$ quantum system, we present a construction of strongly nonlocal quantum states containing $6(d-1)^2$ orthogonal product states, which is one order of magnitude less than the number of basis states $d^3$. Secondly, we give the explicit form of strongly nonlocal orthogonal product basis in $\mathbb{C}^3\otimes \mathbb{C}^3\otimes \mathbb{C}^3\otimes \mathbb{C}^3$ quantum system, where four is the largest known number of subsystems in which there exists strong quantum nonlocality up to now. Both the two results positively answer the open problems in [Halder, \textit{et al.}, PRL, 122, 040403 (2019)], that is, there do exist and even smaller number of quantum states can demonstrate strong quantum nonlocality without entanglement.
\end{abstract}

\maketitle


\section{Introduction}

Quantum nonlocality, as one fundamental property of quantum mechanics, has always received widespread consideration. Generally speaking, quantum entanglement and the violation of Bell inequality are the usual evidences to show the existence of quantum nonlocality. But in recent years, the local discrimination of quantum states, see \cite{bennett1999quantum,RN1380,RN1129,RN1377,RN1139,RN4268,RN4275,RN1133,RN1238,RN567,RN1157,chitambar2013local,RN1966,liu2019distinguishing,duan2010locally} for an incomplete list, has been widely used to assure quantum nonlocality. That means, if the shared quantum state, which is chosen from a known set of quantum states, cannot be distinguished by Alice and Bob using local operations and classical communication, there will be quantum nonlocality existed. Especially the local indiscrimination for product states presents the phenomenon, i.e., quantum nonlocality without entanglement.

Actually until last year, almost all the above references focus on the local discrimination of bipartite quantum states, entangled or not. For tripartite quantum states, Halder \textit{et al.} came up with strong quantum nonlocality without entanglement \cite{halder2019strong}, and presented two explicit strongly nonlocal sets of quantum states in $\mathbb{C}^3 \otimes \mathbb{C}^3 \otimes \mathbb{C}^3$ and $\mathbb{C}^4 \otimes \mathbb{C}^4 \otimes \mathbb{C}^4$ quantum system, respectively. In their opinion, there exists strong quantum nonlocality if for tripartite quantum states, they are locally irreducible in every bipartition. This definition is natural and reasonable, but it seems not so exhaustive for more than tripartite quantum states. In ref. \cite{zhang2019strong}, Zhang \textit{et al.} gave a more general definition of strong quantum nonlocality for multipartite quantum states, that is, ``In $\mathbb{C}^{d_1} \otimes \mathbb{C}^{d_2} \otimes \cdots \otimes \mathbb{C}^{d_n}, \, n \leq 3$, a set of orthogonal quantum states is strongly nonlocal if it is locally irreducible in every $(n-1)$-partition, where $(n-1)$-partition means the whole quantum system is divided into $n-1$ parts.'' Just as they have discussed in \cite{zhang2019strong}, the strength of nonlocality of a set of states which is locally irreducible in every $j$-partition is more than another set of states which is locally irreducible in every $(j+1)$-partition. In other words, if a set of multipartite quantum states is irreducible in every bipartition, then this set is also irreducible in every more than two partition. That means, the quantum nonlocality of a set of quantum states which is locally irreducible in every bipartition is the strongest. Actually this is exactly the original definition of ``strong quantum nonlocality'' in \cite{halder2019strong}, and what we will discussed is this strong quantum nonlocality of multipartite quantum product states.\\

As known, there have already existed two bases in tripartite quantum systems, specifically $\mathbb{C}^3 \otimes \mathbb{C}^3 \otimes \mathbb{C}^3$ and $\mathbb{C}^4 \otimes \mathbb{C}^4 \otimes \mathbb{C}^4$ quantum system respectively, which are strongly nonlocal, but the construction for general tripartite system is still unknown. In \cite{zhang2019strong}, Zhang \textit{et al.} have presented a set of four-partite product states which is locally irreducible in every tripartition, but they do not explain whether it is strongly nonlocal or not in our definition. Thus the two open problems in \cite{halder2019strong} are still open, those are,
\begin{itemize}
	\item[(1)] Can we construct strongly nonlocal quantum product states, even incomplete orthogonal product basis, in general $\mathbb{C}^d \otimes \mathbb{C}^d \otimes \mathbb{C}^d$ quantum system?
	\item[(2)] Can we construct strongly  nonlocal quantum product states in more than tripartite quantum systems?
\end{itemize}

Fortunately in this paper, we solve these two questions and give positive answers.
Firstly, we build $6(d-1)^2$ strongly nonlocal quantum product states in general $\mathbb{C}^d\otimes\mathbb{C}^d\otimes\mathbb{C}^d$ system, which is one order of magnitude less than the whole dimension $d^3$. Specifically in $\mathbb{C}^3\otimes\mathbb{C}^3\otimes\mathbb{C}^3$, $\mathbb{C}^4\otimes\mathbb{C}^4\otimes\mathbb{C}^4$, the number decrease from the original 27, 64 to 24, 54. And obviously, as the dimension increasing, the number will decrease more and more. Further, these are the first incomplete orthogonal product bases which are strongly nonlocal. Secondly, we study the strongest quantum nonlocality in four-party systems. And in $\mathbb{C}^3 \otimes \mathbb{C}^3 \otimes \mathbb{C}^3 \otimes \mathbb{C}^3$ quantum system, we are able to present the explicit form of truly strongly nonlocal sets of  orthogonal quantum product states.

The rest of this paper is organized as follows. In Section \ref{sec:pre}, we present two necessary definitions of strong nonlocality.
In Section \ref{sec:333_nonlocal}, we propose the strongly nonlocal quantum state sets in $\mathbb{C}^d\otimes\mathbb{C}^d\otimes\mathbb{C}^d$ quantum system.
In Section \ref{sec:3333_nonlocal}, we show two sets of strongly nonlocal quantum state sets in
$\mathbb{C}^3\otimes\mathbb{C}^3\otimes\mathbb{C}^3\otimes\mathbb{C}^3$ quantum system.
Finally, we summarize in Section \ref{sec:summary}.

\section{Preliminaries}
\label{sec:pre}
In this section, we will review two definitions which are used through the following sections. The first is the definition of local irreducibility, based on which we can explain the second definition of strong nonlocality.

\begin{definition}[Local irreducibility, \cite{halder2019strong}]\label{def:Locally_irreducible}
A set of orthogonal quantum states is locally irreducible if it is not possible to eliminate one or more quantum states from the set by nontrivial orthogonality-preserving local measurements.
\end{definition}

The so-called ``measurement'' through our paper is a POVM measurement if there is no special explanation. A measurement is nontrivial if at least one POVM element is not proportional to the identity operator.
Otherwise, we call it a trivial measurement. Based on Definition \ref{def:Locally_irreducible}, now we can review the formal definition of strong nonlocality in \cite{halder2019strong} as below, which is the strongest case in \cite{zhang2019strong}.

\begin{definition}[Strong nonlocality, \cite{halder2019strong}]\label{def:strong_nonlocality}
A set of orthogonal product states $\{|\psi\rangle\in\mathcal{H}=\bigotimes_{i=1}^n\mathcal{H}_i| |\psi_i\rangle=|\alpha_i\rangle_1|\beta_i\rangle_2\cdots|\gamma\rangle_n,~ |\alpha_i\rangle_1\in\mathcal{H}_1,|\beta_i\rangle_2\in\mathcal{H}_2\cdots|\gamma\rangle_n\in\mathcal{H}_n \}$
is strongly nonlocal if it is locally irreducible in every bipartition, where $n\ge 3$ and $\dim \mathcal{H}_i\ge 3$ for $i=1,\ldots,n$.
\end{definition}

Next we will explain the relationships among strong nonlocality, local irreducibility and nonlocality by two detailed examples. The reason why we referred nonlocality is the fact that most researchers have focused on this topic before. Actually the strength of local irreducibility is between strong nonlocality and nonlocality. Firstly, local irreducibility is stronger than nonlocality, that is, a locally irreducible set is locally indistinguishable but not vice versa. The following example is a set which is locally indistinguishable but is locally reducible. Consider a set of orthogonal product states in $\mathbb{C}^6\otimes \mathbb{C}^6$:
\begin{equation*}
    \begin{array}{llll}
   |0\rangle|0\pm1\rangle & |0\pm1\rangle|2\rangle &
   |2\rangle|1\pm2\rangle &
   |1\pm 2\rangle|0\rangle\\
   |3\rangle|3\pm4\rangle & |3\pm4\rangle|5\rangle &
   |5\rangle|4\pm5\rangle &
   |4\pm 5\rangle|3\rangle\\
    \end{array}
\end{equation*}
The above set is locally indistinguishable because a subset of these states is locally indistinguishable \cite{bennett1999quantum}. However, if Alice perform a local measurement $\{\sum_{i=0}^2|i\rangle\langle i|,\sum_{i=3}^5|j\rangle\langle j|\}$, half of above states will be eliminated. Namely, this set of states are locally reducible.

Secondly, a locally irreducible set is not a strongly nonlocal set in general. Consider a set of orthogonal product states in $\mathbb{C}^3\otimes\mathbb{C}^3\otimes\mathbb{C}^3$, or to say $\mathcal{H}_A\otimes\mathcal{H}_B\otimes\mathcal{H}_C$:
\begin{equation*}
    \begin{array}{lll}
   |0\rangle|1\rangle|0\pm1\rangle & |1\rangle|0\pm1\rangle|0\rangle & |0\pm1\rangle|0\rangle|1\rangle\\
   |0\rangle|2\rangle|0\pm2\rangle & |2\rangle|0\pm2\rangle|0\rangle & |0\pm2\rangle|0\rangle|2\rangle\\
   \end{array}
\end{equation*}
The above set is locally irreducible, which was shown in \cite{halder2019strong}. However, if we perform measurement $\{|0\rangle|1\rangle\langle0|\langle1|+|0\rangle|2\rangle\langle0|\langle2|,\mathbb{I}-(|0\rangle|1\rangle\langle0|\langle1|+|0\rangle|2\rangle\langle0|\langle2|)\}$ in a composite space $\mathcal{H}_A\otimes\mathcal{H}_B$, then we can eliminate at least one quantum state. That means, the set is locally reducible in bipartition ``$AB|C$'', which does not satisfy the condition of strong nonlocality.

\section{Strongly nonlocal orthogonal product states in $\mathbb{C}^d\otimes{\mathbb{C}}^d\otimes {\mathbb{C}}^d$}
\label{sec:333_nonlocal}
In this section, we will construct a set of strongly nonlocal orthogonal product states (OPS) in $\mathbb{C}^d\otimes{\mathbb{C}}^d\otimes {\mathbb{C}}^d$, represented as Alice (A), Bob (B) and Charlie (C)'s Hilbert space $\mathcal{H}_A$, $\mathcal{H}_B$ and $\mathcal{H}_C$, with $d\ge 3$. The basis for each subsystem is $\{ |i\rangle \}_{i=0}^{d-1}$, and in our paper $|i\pm j\rangle$ denotes the state $\frac{1}{\sqrt{2}}(|i\rangle\pm|j\rangle)$.

We consider the following OPS in $\mathbb{C}^d\otimes{\mathbb{C}}^d\otimes {\mathbb{C}}^d$:
\begin{equation}\label{ddd_nonlocal}
  \begin{array}{ccc}
   |0\rangle|i\rangle|0\pm i\rangle & |i\rangle|0\pm i\rangle|0\rangle &|0\pm i\rangle|0\rangle|i\rangle,\\
   |i\rangle|j\rangle|0\pm i\rangle & |j\rangle|0\pm i\rangle|i\rangle &|0\pm i\rangle|i\rangle|j\rangle,
  \end{array}
\end{equation}
where $1\le i,j\le d-1$ and $i\neq j$.
Note that the set (\ref{ddd_nonlocal}) is invariant under cyclic permutation of the three parties A, B and C. On the account that local irreducibility is a sufficient condition for strong nonlocality, we first show that these states are locally irreducible in Lemma \ref{lem:ddd_irreducible}.

\begin{lemma}\label{lem:ddd_irreducible}
In $\mathbb{C}^d\otimes\mathbb{C}^d\otimes\mathbb{C}^d$, the above $6(d-1)^2$ states (\ref{ddd_nonlocal}) are locally irreducible.
\end{lemma}
\begin{proof}
Through the definition of local irreducibility, we need to illustrate the impossibility of eliminating one or more quantum states by nontrivial measurements. If there is even no nontrivial measurements existed to preserve orthogonality, then the states set must be locally irreducible.
Consider the subset of OPS (\ref{ddd_nonlocal}):
\begin{equation}\label{ddd_irreducible}
  \begin{array}{c}
   |\alpha_i^{1,2}\rangle=|0\rangle|i\rangle|0\pm i\rangle, \\ |\alpha_i^{3,4}\rangle=|i\rangle|0\pm i\rangle|0\rangle, \\ |\alpha_i^{5,6}\rangle=|0\pm i\rangle|0\rangle|i\rangle,
  \end{array}
\end{equation}
where $1\le i\le d-1$.
It is obvious that if any orthogonality-preserving local measurement turns out to be trivial for the quantum states set (\ref{ddd_irreducible}), then the satisfied measurement for the quantum states set (\ref{ddd_nonlocal}) also must be trivial, so the local irreducibility of the quantum states set (\ref{ddd_nonlocal}) can be proved.

We now show that any orthogonality-preserving local POVM perform either on $A$, $B$ or $C$ must be trivial.  First, we assume Alice goes first. Let POVM $\{\Pi_A\}$ denote a general orthogonality-preserving measurement on $A$. Each POVM element can be written as a $d\times d$ matrix in the $\{|0\rangle,\cdots,|d-1\rangle\}$ basis of $\mathcal{H}_A$ generally:
\[
\Pi_A=\left(
                      \begin{array}{ccc}
                        a_{0,0} & \cdots & a_{0,d-1} \\
                        \vdots & \ddots & \vdots \\
                        a_{d-1,0} & \cdots & a_{d-1,d-1} \\
                      \end{array}
                    \right)\in\mathbb{C}^{d\times d}.
\]
The measurement must leave the postmeasurement states mutually orthogonal.
By setting the the inner product $\langle\alpha_i^1|\Pi_A \otimes \mathbb{I}_B\otimes \mathbb{I}_C|\alpha_i^3\rangle=0$, we obtain $a_{0,i}=a_{i,0}=0$ for $1\le i \le d-1$.
For $1\le i,j\le d-1$ and $i\neq j$, by setting $\langle\alpha_i^3|\Pi_A \otimes \mathbb{I}_B\otimes \mathbb{I}_C|\alpha_j^5\rangle=0$, we have $a_{i,j}=0$.
By setting $\langle\alpha_i^5|\Pi_A \otimes \mathbb{I}_B\otimes \mathbb{I}_C|\alpha_i^6\rangle=0$, we have $a_{0,0}=a_{i,i}$ for $1\le i\le d-1$.
As the diagonal elements of $\Pi_A$ are all equal and the off-diagonal elements are all zero, $\Pi_A$ must be proportional to the identity.
Therefore, Alice can not go first. Since quantum states set (\ref{ddd_irreducible}) are invariant under cyclic permutation of the parties, Bob and Charlie can not go first either. This completes the proof of the lemma.
\end{proof}

Lemma \ref{lem:ddd_irreducible} shows the local irreducibility of states (\ref{ddd_nonlocal}), which also implies quantum nonlocality. Next, we sill show their local irreducibility in every bipartition, which can assure the strong nonlocality of (\ref{ddd_nonlocal}) in Theorem \ref{thm:ddd_nonlocal}.

\begin{theorem}\label{thm:ddd_nonlocal}
In $\mathbb{C}^d\otimes\mathbb{C}^d\otimes\mathbb{C}^d$, the above quantum states  (\ref{ddd_nonlocal}) are strongly nonlocal.
\end{theorem}
\begin{proof}
Since quantum state set (\ref{ddd_nonlocal}) is invariant under cyclic permutation of the parties, we only need to consider the bipartition $A|BC$. We will show that any orthogonality-preserving local POVM performed either on A or BC together must be trivial.
In bipartition $A|BC$, the state set (\ref{ddd_nonlocal}) takes the form:
\begin{equation*}
  \begin{array}{ll}
   |\alpha_i^{1,2}\rangle=|0\rangle|i0\pm ii\rangle & |\beta_{ij}^{1,2}\rangle=|i\rangle|j0\pm ji\rangle \\
   |\alpha_i^{3,4}\rangle=|i\rangle|00\pm i0\rangle & |\beta_{ij}^{3,4}\rangle=|j\rangle|0i\pm ii\rangle \\
   |\alpha_i^{5,6}\rangle=|0\pm i\rangle|0i\rangle & |\beta_{ij}^{5,6}\rangle=|0\pm i\rangle|ij\rangle\\
  \end{array}
\end{equation*}
where $1\le i,j\le d-1$ and $i\neq j$.

According to Lemma \ref{lem:ddd_irreducible}, Alice can not go first. We consider whether it is possible to initial a local protocol by performing some non-trivial POVM on $BC$. Let $\{\Pi_{BC}\}$ be a general orthogonal-preserving measurement on $BC$. Each POVM element can be written as a $d^2\times d^2$ matrix in the $\{|00\rangle,|01\rangle, \cdots, |(d-1)(d-1)\rangle$ basis of $\mathcal{H}_{BC}$. Assume one element in $\{ \Pi_{BC}  \}$ is as follows.
\[
\Pi_{BC}=\left(
                     \begin{array}{ccc}
                       b_{00,00} & \cdots & b_{00,(d-1)(d-1)} \\
                       \vdots & \ddots & \vdots \\
                       b_{(d-1)(d-1),00} & \cdots & b_{(d-1)(d-1),(d-1)(d-1)} \\
                     \end{array}
                   \right)
\]
For states $|\alpha_i^{1,2}\rangle$, $|\alpha_j^{1,2}\rangle$ and $1\le i\neq j\le d-1$, we known $\langle\alpha_i^{1,2}|\mathbb{I}_A\otimes\Pi_{BC}|\alpha_j^{1,2}\rangle=0$, which implies $b_{i0,j0}=b_{i0,jj}=b_{ii,j0}=b_{ii,jj}=0$.
For states $|\alpha_i^{5}\rangle$, $|\alpha_j^{5}\rangle$ and $1\le i\neq j\le d-1$, we have $\langle\alpha_i^{5}|\mathbb{I}_A\otimes\Pi_{BC}|\alpha_j^{5}\rangle=0$, \emph{i.e.}, $b_{0i,0j}=0$.
For states $|\beta_{st}^5\rangle$, $|\beta_{st}^5\rangle$, $1\le i,j,s,t\le d-1$, $s\ne t$ and $i\ne j$, we have $\langle\beta_{st}^{5}|\mathbb{I}_A\otimes\Pi_{BC}|\beta_{ij}^{5}\rangle=0$, which implies $b_{st,ij}=0$.
For $1\le i,j\le d-1$, if $\langle\alpha_i^{1,2}|\mathbb{I}_A\otimes\Pi_{BC}|\alpha_j^{5}\rangle=0$, then we have $b_{i0,0j}=b_{ii,0j}=b_{0j,i0}=b_{0j,ii}=0$.
For $1\le i,s,t\le d-1$ and $s\ne t$, if $\langle\alpha_i^{1,2}|\mathbb{I}_A\otimes\Pi_{BC}|\beta_{st}^{5}\rangle=0$, then we have $b_{i0,st}=b_{ii,st}=b_{st,i0}=b_{st,ii}=0$.
For states $|\alpha_{i}^{3,4}\rangle$, $|\alpha_i^{5}\rangle$ and $1\le i\le d-1$, we have $\langle\alpha_i^{3,4}|\mathbb{I}_A\otimes\Pi_{BC}|\alpha_{i}^{5}\rangle=0$, which results in $b_{00,0i}=b_{i0,0i}=b_{0i,00}=b_{0i,i0}=0$.
For states $|\alpha_{i}^{3,4}\rangle$, $\beta_{ij}^{1,2}$ and $1\le i\neq j\le d-1$, we obtain $\langle\alpha_i^{3,4}|\mathbb{I}_A\otimes\Pi_{BC}|\beta_{ij}^{1,2}\rangle=0$, which means $b_{00,j0}=b_{00,ji}=b_{j0,00}=b_{ji,00}=0$.
For $1\le i\ne j\le d-1$, if $\langle\alpha_i^{3,4}|\mathbb{I}_A\otimes\Pi_{BC}|\beta_{ji}^{3,4}\rangle=0$, then we obtain $b_{00,jj}=b_{jj,00}=0$.
For $1\le i,s,t\le d-1$ and $s\ne t$, if $\langle\alpha_i^{5}|\mathbb{I}_A\otimes\Pi_{BC}|\beta_{st}^{5}\rangle=0$, then we have $b_{0i,st}=b_{st,0i}=0$.
For states $|\beta_{is}^{1,2}\rangle$, $|\beta_{ti}^{3,4}\rangle$, $1\le i,s,t\le d-1$, $i\neq s$ and $i\neq t$, we have $\langle\beta_{is}^{1,2}|\mathbb{I}_A\otimes\Pi_{BC}|\beta_{ti}^{3,4}\rangle=0$, which implies $b_{s0,0t}=b_{s0,tt}=b_{0t,s0}=b_{tt,s0}=0$.
All above discussion illuminates that all off-diagonal elements of $\Pi_{AB}$ are equal to 0.

Moreover, for state $|\alpha_{i}^{1,2}\rangle$ and $1\le i\le d-1$, let $\langle\alpha_i^{1}|\mathbb{I}_A\otimes\Pi_{BC}|\alpha_{i}^{2}\rangle=0$, we can obtain $b_{i0,i0}=b_{ii,ii}$.
For state $|\alpha_{i}^{3,4}\rangle$ and $1\le i\le d-1$, let $\langle\alpha_i^{3}|\mathbb{I}_A\otimes\Pi_{BC}|\alpha_{i}^{4}\rangle=0$, we can obtain $b_{00,00}=b_{i0,i0}$.
For state $|\beta_{ij}^{1,2}\rangle$ and $1\le i\ne j\le d-1$, let $\langle\beta_{ij}^{1}|\mathbb{I}_A\otimes\Pi_{BC}|\beta_{ij}^{2}\rangle=0$, we can obtain $b_{j0,j0}=b_{ji,ji}$.
For state $|\beta_{ij}^{3,4}\rangle$ and $1\le i\ne j\le d-1$, let $\langle\beta_{ij}^{3}|\mathbb{I}_A\otimes\Pi_{BC}|\beta_{ij}^{4}\rangle=0$, we can obtain $b_{0i,0i}=b_{ii,ii}$.
That is, all the diagonal elements of $\Pi_{BC}$ are equal.

As the diagonal elements of $\Pi_{BC}$ are all equal and the off-diagonal elements are all zero, $\Pi_{BC}$ must be proportional to the identity. Hence, $BC$ can not go first. This completes the proof of theorem.
\end{proof}

Up to now, we have constructed $6(d-1)^2$ strongly nonlocal OPS in $\mathbb{C}^d\otimes\mathbb{C}^d\otimes\mathbb{C}^d$ quantum system. The number of constructed states is one order of magnitude less than the dimension of whole space $d^3$, which is the number of existed strongly nonlocal product states. Especially, we present 24 and 54 strongly nonlocal OPB in $\mathbb{C}^3\otimes\mathbb{C}^3\otimes\mathbb{C}^3$ and $\mathbb{C}^4\otimes\mathbb{C}^4\otimes\mathbb{C}^4$, which is 3 and 10 fewer than states proposed in \cite{halder2019strong} respectively. This gives a positive answer to one open problem in \cite{halder2019strong} that ``whether incomplete orthogonal product bases can be strongly nonlocal''. Theorem \ref{thm:ddd_nonlocal} also illuminates that local implementation of a three-party separable measurement in general three-party space can not be realized unless uses entanglement across all bipartitions.

\section{Strongly Nonlocal Orthogonal Product States in general system $\mathbb{C}^3\otimes \mathbb{C}^3\otimes \mathbb{C}^3\otimes \mathbb{C}^3$}
\label{sec:3333_nonlocal}
Ref. \cite{halder2019strong} leaves another open problem that how to construct strongly nonlocal OPS in multiparty systems? In \cite{zhang2019strong}, a set of `strongly nonlocal' OPS was proposed in a specific four-party system. However, the presented set is not a strongly nonlocal set defined in this paper. Since only part of bipartitions are consider in \cite{zhang2019strong}, it is a partial answer for this open question.

In this section, we will construct a set of strongly nonlocal orthogonal product states in specific four-party systems: $\mathbb{C}^3\otimes \mathbb{C}^3\otimes \mathbb{C}^3\otimes \mathbb{C}^3$ and $\mathbb{C}^4\otimes \mathbb{C}^4\otimes \mathbb{C}^4\otimes \mathbb{C}^4$. The strong nonlocality has not been considered in previous works.
Here similarly, let $|i\pm j\rangle|k\pm l\rangle$ denote quantum states $\frac{1}{2}(|i\rangle+|j\rangle)(|k\rangle+|l\rangle)$, $\frac{1}{2}(|i\rangle+|j\rangle)(|k\rangle-|l\rangle)$, $\frac{1}{2}(|i\rangle-|j\rangle)(|k\rangle+|l\rangle)$ and $\frac{1}{2}(|i\rangle-|j\rangle)(|k\rangle-|l\rangle)$.

The strongly nonlocal orthogonal product basis (OPB) in $\mathbb{C}^3\otimes \mathbb{C}^3\otimes \mathbb{C}^3\otimes \mathbb{C}^3$ is described as follows:

\begin{equation}\label{3333_nonlocal}
\begin{array}{lll}
    |0\rangle|1\rangle|0\pm1\rangle|0\pm2\rangle & |1\rangle|1\rangle|0\rangle|1\pm 2\rangle & |1\rangle|0\rangle|1\rangle|0\rangle\\
    |1\rangle|2\rangle|1\pm2\rangle|1\pm0\rangle & |2\rangle|2\rangle|1\rangle|2\pm 0\rangle & |2\rangle|1\rangle|2\rangle|1\rangle\\
    |2\rangle|0\rangle|2\pm 0\rangle|2\pm 1\rangle & |0\rangle|0\rangle|2\rangle|0\pm 1\rangle & |0\rangle|2\rangle|0\rangle|2\rangle\\
    |1\rangle|0\pm1\rangle|0\pm2\rangle|0\rangle & |1\rangle|0\rangle|1\pm 2\rangle|1\rangle & |0\rangle|1\rangle|0\rangle|1\rangle\\
    |2\rangle|1\pm2\rangle|1\pm0\rangle|1\rangle & |2\rangle|1\rangle|2\pm 0\rangle|2\rangle & |1\rangle|2\rangle|1\rangle|2\rangle\\
    |0\rangle|2\pm 0\rangle|2\pm 1\rangle|2\rangle & |0\rangle|2\rangle|0\pm 1\rangle|0\rangle & |2\rangle|0\rangle|2\rangle|0\rangle\\
    |0\pm1\rangle|0\pm2\rangle|0\rangle|1\rangle & |0\rangle|1\pm 2\rangle|1\rangle|1\rangle & |0\rangle|0\rangle|0\rangle|0\rangle\\
    |1\pm2\rangle|1\pm0\rangle|1\rangle|2\rangle & |1\rangle|2\pm 0\rangle|2\rangle|2\rangle & |1\rangle|1\rangle|1\rangle|1\rangle\\
    |2\pm 0\rangle|2\pm 1\rangle|2\rangle|0\rangle & |2\rangle|0\pm 1\rangle|0\rangle|0\rangle & |2\rangle|2\rangle|2\rangle|2\rangle\\
    |0\pm2\rangle|0\rangle|1\rangle|0\pm1\rangle & |1\pm 2\rangle|1\rangle|1\rangle|0\rangle &\\
    |1\pm0\rangle|1\rangle|2\rangle|1\pm2\rangle & |2\pm 0\rangle|2\rangle|2\rangle|1\rangle &\\
    |2\pm 1\rangle|2\rangle|0\rangle|2\pm 0\rangle & |0\pm 1\rangle|0\rangle|0\rangle|2\rangle &\\
\end{array}
\end{equation}

Firstly, we will show states (\ref{3333_nonlocal}) are locally irreducible.
\begin{lemma}\label{lem:3333_irreducible}
In $\mathbb{C}^3\otimes\mathbb{C}^3\otimes\mathbb{C}^3\otimes\mathbb{C}^3$, quantum states (\ref{3333_nonlocal}) are locally irreducible.
\end{lemma}
\begin{proof}
Consider the subset of OPB (\ref{3333_nonlocal}):
\begin{equation}\label{3333_irreducible}
\begin{array}{l l}
|\psi_{1,2}\rangle=|1\rangle|1\rangle|0\rangle|1\pm 2\rangle & |\psi_{13,14}\rangle=|1\rangle|0\rangle|1\pm 2\rangle|1\rangle \\
|\psi_{3,4}\rangle=|2\rangle|2\rangle|1\rangle|2\pm 0\rangle & |\psi_{15,16}\rangle=|2\rangle|1\rangle|2\pm 0\rangle|2\rangle\\
|\psi_{5,6}\rangle=|0\rangle|0\rangle|2\rangle|0\pm 1\rangle & |\psi_{17,18}\rangle=|0\rangle|2\rangle|0\pm 1\rangle|0\rangle\\
|\psi_{7,8}\rangle=|0\rangle|1\pm 2\rangle|1\rangle|1\rangle & |\psi_{19,20}\rangle=|1\pm 2\rangle|1\rangle|1\rangle|0\rangle\\
|\psi_{9,10}\rangle=|1\rangle|2\pm 0\rangle|2\rangle|2\rangle & |\psi_{21,22}\rangle=|2\pm 0\rangle|2\rangle|2\rangle|1\rangle\\
|\psi_{11,12}\rangle=|2\rangle|0\pm 1\rangle|0\rangle|0\rangle & |\psi_{23,24}\rangle=|0\pm 1\rangle|0\rangle|0\rangle|2\rangle\\
\end{array}
\end{equation}

We will show that any orthogonality-preserving local POVM perform either on $A$, $B$, $C$ or $D$ must be trivial.
Let the POVM $\{\Pi_A\}$ be a general orthogonal-preserving measurement on $A$.
Each POVM can be written as a $3\times3$ matrix on the $\{|0\rangle,|1\rangle,|2\rangle\}$ basis of $\mathcal{H}_A$:
\begin{equation*}
\Pi_A=\left(
  \begin{array}{ccc}
    a_{0,0} & a_{0,1} & a_{0,2} \\
    a_{1,0} & a_{1,1} & a_{1,2} \\
    a_{2,0} & a_{2,1} & a_{2,2} \\
  \end{array}
\right)\in\mathbb{C}^{3\times 3}.
\end{equation*}
If Alice's measurement leaves the states orthogonal, then for $|\psi_{5}\rangle$ and $|\psi_{13}\rangle$,
$\langle\psi_{5}|\Pi_A\otimes\mathbb{I}_B\otimes\mathbb{I}_C\otimes\mathbb{I}_D|\psi_{13}\rangle=0$. Thus, $a_{0,1}=a_{1,0}=0$.
Similarly, for $|\psi_{3}\rangle, |\psi_{17}\rangle$ and $|\psi_{1}\rangle, |\psi_{15}\rangle$, we can obtain
$a_{0,2}=a_{2,0}=0$ and $a_{1,2}=a_{2,1}=0$ respectively. For $|\psi_{19}\rangle$ and $|\psi_{20}\rangle$, $\langle\psi_{19}|\Pi_A\otimes\mathbb{I}_B\otimes\mathbb{I}_C\otimes\mathbb{I}_D|\psi_{20}\rangle=0$ implies $a_{1,1}=a_{2,2}$.
Similarly, for $|\psi_{21}\rangle$ and $|\psi_{22}\rangle$, we can obtain $a_{2,2}=a_{0,0}$. Namely, if Alice goes first, the POVM element $\Pi_A$ is
proportional to the identity, which is a trivial POVM element. Hence, Alice can not go first. Since the quantum states (\ref{3333_irreducible}) are invariant under the cyclic permutation of the parties, Bob, Charlie and Eve can not go first.
This completes the proof of the lemma.
\end{proof}

\begin{lemma}\label{lem:3333_irreducible2}
In every bipartition, i.e., $AB|CD$, $AC|BD$ and $AD|BC$, orthogonal product states (\ref{3333_nonlocal}) are locally irreducible.
\end{lemma}
\begin{proof}
Consider the subset of OPB (\ref{3333_nonlocal}):
\begin{equation}
\begin{array}{l l}\label{3333_irreducible2}
    |\phi_{1,2,3,4}\rangle=|0\rangle|1\rangle|0\pm1\rangle|0\pm2\rangle &  \\
    |\phi_{5,6,7,8}\rangle=|1\rangle|2\rangle|1\pm2\rangle|1\pm0\rangle &  \\
    |\phi_{9,10,11,12}\rangle=|2\rangle|0\rangle|2\pm 0\rangle|2\pm 1\rangle &  \\
    |\phi_{13,14,15,16}\rangle=|1\rangle|0\pm1\rangle|0\pm2\rangle|0\rangle &  \\
    |\phi_{17,18,19,20}\rangle=|2\rangle|1\pm2\rangle|1\pm0\rangle|1\rangle &   \\
    |\phi_{21,22,23,24}\rangle=|0\rangle|2\pm 0\rangle|2\pm 1\rangle|2\rangle &   \\
    |\phi_{25,26,27,28}\rangle=|0\pm1\rangle|0\pm2\rangle|0\rangle|1\rangle &  \\
    |\phi_{29,30,31,32}\rangle=|1\pm2\rangle|1\pm0\rangle|1\rangle|2\rangle &  \\
    |\phi_{33,34,35,36}\rangle=|2\pm 0\rangle|2\pm 1\rangle|2\rangle|0\rangle &  \\
    |\phi_{37,38,39,40}\rangle=|0\pm2\rangle|0\rangle|1\rangle|0\pm1\rangle &  \\
    |\phi_{41,42,43,44}\rangle=|1\pm0\rangle|1\rangle|2\rangle|1\pm2\rangle &  \\
    |\phi_{45,46,47,48}\rangle=|2\pm 1\rangle|2\rangle|0\rangle|2\pm 0\rangle &  \\
\end{array}
\end{equation}
We will show quantum states (\ref{3333_irreducible}) and (\ref{3333_irreducible2}) are irreducible in bipartition $AB|CD$.
If Alice and Bob go first, let POVM $\{\Pi_{AB}\}$ be a general orthogonal-preserving measure on $\mathcal{H}_A\otimes\mathcal{H}_B$.
Each POVM element can be written as a $9\times 9$ matrix on basis $\{|ij\rangle|i,j=0,1,2\}$ of $\mathcal{H}_A\otimes\mathcal{H}_B$:
\[
\Pi_{AB}=\left(
           \begin{array}{ccc}
             b_{00,00} & \cdots & b_{00,22} \\
             \vdots & \ddots  & \vdots \\
             b_{22,00} & \cdots  & b_{22,22} \\
           \end{array}
         \right)\in\mathbb{C}^{9\times 9}.
\]

For $|\psi_5\rangle$ and $|\psi_{21,22}\rangle$, let $\langle\psi_5|\Pi_{AB}\otimes\mathbb{I}_C\otimes\mathbb{I}_D|\psi_{21,22}\rangle=0$. We have $a_{00,22}\pm a_{00,02}=0$ which implies $a_{00,22}=a_{00,02}=0$.
For $|\phi_{17,19}\rangle$ and $|\phi_{25,26,27,28}\rangle$, let $\langle\phi_{17,19}|\Pi_{AB}\otimes\mathbb{I}_C\otimes\mathbb{I}_D|\phi_{25,26,27,28}\rangle=0$.
We have
\begin{gather*}
    (a_{21,00}+a_{21,02}+a_{21,10}+a_{21,12})\pm \\
   (a_{22,00}+a_{22,02}+a_{22,10}+a_{22,12})=0 \\
    (a_{21,00}-a_{21,02}+a_{21,10}-a_{21,12})\pm\\
    (a_{22,00}-a_{22,02}+a_{22,10}-a_{22,12})=0 \\
    (a_{21,00}+a_{21,02}-a_{21,10}-a_{21,12})\pm \\
    (a_{22,00}+a_{22,02}-a_{22,10}-a_{22,12})=0 \\
    (a_{21,00}-a_{21,02}-a_{21,10}+a_{21,12})\pm \\
    (a_{22,00}-a_{22,02}-a_{22,10}+a_{22,12})=0
\end{gather*}
which implies $a_{21,00}=a_{21,02}=a_{21,10}=a_{21,12}=a_{22,00}=a_{22,02}=a_{22,10}=a_{22,12}=0$.
Similarly, we can obtain the off-diagonal elements of $\Pi_{AB}$ are all equal to 0.

Let $\langle \phi_{25}|\Pi_{AB}\otimes\mathbb{I}_C\otimes\mathbb{I}_D|\phi_{26,27,28}\rangle=0$, we can obtain $a_{00,00}=a_{02,02}=a_{10,10}=a_{12,12}$.
Let $\langle \phi_{29}|\Pi_{AB}\otimes\mathbb{I}_C\otimes\mathbb{I}_D|\phi_{30,31,32}\rangle=0$, we can obtain $a_{11,11}=a_{10,10}=a_{21,21}=a_{20,20}$.
Let $\langle \phi_{33}|\Pi_{AB}\otimes\mathbb{I}_C\otimes\mathbb{I}_D|\phi_{34,36,36}\rangle=0$, we can obtain $a_{22,22}=a_{21,21}=a_{02,02}=a_{01,01}$.
Since all the diagonal elements of $\Pi_{AB}$ are equal, $\Pi_{AB}$ is proportional to an identity. Namely, Alice and Bob can not go first.
Since quantum states (\ref{3333_irreducible}) and (\ref{3333_irreducible2}) are invariant under permutation, Charlie and Eve can not go first.
That is, under the partition $AB|CD$, quantum states (\ref{3333_irreducible}) and (\ref{3333_irreducible2}) are locally irreducible.
Similarly, (\ref{3333_irreducible}) and (\ref{3333_irreducible2}) are locally irreducible under partition $AC|BD$ and $AD|BC$.
Find full proof of Lemma \ref{lem:3333_irreducible2} in Appendix \ref{app:lem3333irreducible2}.
\end{proof}

\begin{lemma}\label{lem:3333_irreducible3}
In every bipartition, i.e., $A|BCD$, $B|ACD$, $C|ABD$ and $D|ABC$, orthogonal product states (\ref{3333_nonlocal}) are locally irreducible.
\end{lemma}
\begin{proof}
Consider the subset of (\ref{3333_nonlocal}):
\begin{equation}\label{3333_irreducible3}
\begin{array}{ccc}
  |\varphi_{1}\rangle=|1\rangle|0\rangle|1\rangle|0\rangle & |\varphi_{4}\rangle=|0\rangle|1\rangle|0\rangle|1\rangle & |\varphi_{7}\rangle=|0\rangle|0\rangle|0\rangle|0\rangle \\
  |\varphi_{2}\rangle=|2\rangle|1\rangle|2\rangle|1\rangle & |\varphi_{5}\rangle=|1\rangle|2\rangle|1\rangle|2\rangle & |\varphi_{8}\rangle=|1\rangle|1\rangle|1\rangle|1\rangle \\
  |\varphi_{3}\rangle=|0\rangle|2\rangle|0\rangle|2\rangle & |\varphi_{6}\rangle=|2\rangle|0\rangle|2\rangle|0\rangle & |\varphi_{9}\rangle=|2\rangle|2\rangle|2\rangle|2\rangle
\end{array}
\end{equation}
We will show quantum states (\ref{3333_nonlocal}) (i.e., (\ref{3333_irreducible}), (\ref{3333_irreducible2}) and (\ref{3333_irreducible3}) ) are irreducible in partition $A|BCD$, $B|ACD$, $C|ABD$ and $D|ABC$.
We only consider the partition $A|BCD$ because (\ref{3333_nonlocal}) is invariant under permutation.
Since (\ref{3333_nonlocal}) is irreducible due to Lemma \ref{lem:3333_irreducible}, Alice can not go first.
Assume that Bob, Charlie and Eve go first. Let $\{\Pi_{BCD}\}$ denote the POVM in composite Hilbert space $\mathcal{H}_B\otimes\mathcal{H}_B\otimes\mathcal{H}_D$.
Each POVM element can be written as a $27\times27$ matrix on the basis $\{|ijk\rangle|i,j,k\in{0,1,2}\}$:
\[
\Pi_{BCD}=\left(
           \begin{array}{ccc}
             c_{000,000} & \cdots & c_{000,222} \\
             \vdots & \ddots  & \vdots \\
             c_{222,000} & \cdots  & c_{222,222} \\
           \end{array}
         \right)\in\mathbb{C}^{27\times 27}.
\]

If measurement leaves the states orthogonal, then for $|\psi_{23}\rangle$ and $|\varphi_7\rangle$, $\langle\psi_{23}|\mathbb{I}_A\otimes\Pi_{BCD}|\varphi_{7}\rangle=0$.
Hence, $c_{022,000}=0$. For $|\phi_{1,2,3,4}\rangle$ and $|\phi_{21,22,23,24}\rangle$, let $\langle\phi_{1,2,3,4}|\mathbb{I}_A\otimes\Pi_{BCD}\otimes|\phi_{21,22,23,25}\rangle=0$, which implies
\begin{gather*}
(c_{100,222}+c_{100,212}+c_{100,022}+c_{100,012})\pm\\
(c_{102,222}+c_{102,212}+c_{102,022}+c_{102,012})\pm\\
(c_{110,222}+c_{110,212}+c_{110,022}+c_{10,012})+\\
(c_{112,222}+c_{112,212}+c_{112,022}+c_{112,012})=0\\
(c_{100,222}+c_{100,212}+c_{100,022}+c_{100,012})\pm\\
(c_{102,222}+c_{102,212}+c_{102,022}+c_{102,012})\mp\\
(c_{110,222}+c_{110,212}+c_{110,022}+c_{10,012})-\\
(c_{112,222}+c_{112,212}+c_{112,022}+c_{112,012})=0\\
(c_{100,222}-c_{100,212}+c_{100,022}-c_{100,012})\pm\\
(c_{102,222}-c_{102,212}+c_{102,022}-c_{102,012})\pm\\
(c_{110,222}-c_{110,212}+c_{110,022}-c_{10,012})+\\
(c_{112,222}-c_{112,212}+c_{112,022}-c_{112,012})=0\\
(c_{100,222}-c_{100,212}+c_{100,022}-c_{100,012})\pm\\
(c_{102,222}-c_{102,212}+c_{102,022}-c_{102,012})\mp\\
(c_{110,222}-c_{110,212}+c_{110,022}-c_{10,012})-\\
(c_{112,222}-c_{112,212}+c_{112,022}-c_{112,012})=0\\
(c_{100,222}+c_{100,212}-c_{100,022}-c_{100,012})\pm\\
(c_{102,222}+c_{102,212}-c_{102,022}-c_{102,012})\pm\\
(c_{110,222}+c_{110,212}-c_{110,022}-c_{10,012})+\\
(c_{112,222}+c_{112,212}-c_{112,022}-c_{112,012})=0\\
(c_{100,222}+c_{100,212}-c_{100,022}-c_{100,012})\pm\\
(c_{102,222}+c_{102,212}-c_{102,022}-c_{102,012})\mp\\
(c_{110,222}+c_{110,212}-c_{110,022}-c_{10,012})-\\
(c_{112,222}+c_{112,212}-c_{112,022}-c_{112,012})=0\\
(c_{100,222}-c_{100,212}-c_{100,022}+c_{100,012})\pm\\
(c_{102,222}-c_{102,212}-c_{102,022}+c_{102,012})\pm\\
(c_{110,222}-c_{110,212}-c_{110,022}+c_{10,012})+\\
(c_{112,222}-c_{112,212}-c_{112,022}+c_{112,012})=0\\
(c_{100,222}-c_{100,212}-c_{100,022}+c_{100,012})\pm\\
(c_{102,222}-c_{102,212}-c_{102,022}+c_{102,012})\mp\\
(c_{110,222}-c_{110,212}-c_{110,022}+c_{10,012})-\\
(c_{112,222}-c_{112,212}-c_{112,022}+c_{112,012})=0\\
\end{gather*}
We have $c_{100,222}=c_{100,212}=c_{100,022}=c_{100,012}=c_{102,222}=c_{102,212}=c_{102,022}=c_{102,012}=c_{110,222}=c_{110,212}=c_{110,022}=c_{10,012}=c_{112,222}=c_{112,212}=c_{112,022}=c_{112,012}=0$.
Similarly, we can prove that all the off-diagonal elements in $\Pi_{BCD}$ are equal to 0.

Let $\langle\phi_{5}|\mathbb{I}_{A}\otimes\Pi_{BCD}|\phi_{6,7,8}\rangle=0$, so that we can obtain $c_{211,211}=c_{210,210}=c_{221,221}=c_{220,220}$.
Analogously, the diagonal elements of $\Pi_{BCD}$ are all equal. See full proof in Appendix \ref{app:lem3333irreducible3}.
Since $\Pi_{BCD}$ is proportional to identity, Bob, Charlie and Eve can not go first.
Namely, quantum states (\ref{3333_nonlocal}) are locally irreducible in bipartition $A|BCD$.
Quantum states (\ref{3333_nonlocal}) are also locally irreducible in bipartition $B|ACD$, $C|ABD$ and $D|ABC$ because (\ref{3333_nonlocal}) are invariant under permutation.
This completes the proof.
\end{proof}

Combine Lemma \ref{lem:3333_irreducible2} and Lemma \ref{lem:3333_irreducible3}, we obtain the strong nonlocality of orthogonal product basis (\ref{3333_nonlocal}), which is exactly the following Theorem \ref{thm:3333_nonlocal}.
\begin{theorem}\label{thm:3333_nonlocal}
In $\mathbb{C}^3\otimes\mathbb{C}^3\otimes\mathbb{C}^3\otimes\mathbb{C}^3$, quantum states (\ref{3333_nonlocal}) are strongly nonlocal.
\end{theorem}

\section{Summary}
\label{sec:summary}
In summary, we have extended the previous known constructions of strongly nonlocal sets to $6(d-1)^2$-size sets in general three-party systems $\mathbb{C}^d\otimes\mathbb{C}^d\otimes\mathbb{C}^d$, which decrease one order of magnitude from the whole dimension $d^3$. Specifically in  $\mathbb{C}^3\otimes\mathbb{C}^3\otimes\mathbb{C}^3$ and $\mathbb{C}^4\otimes\mathbb{C}^4\otimes\mathbb{C}^4$ quantum system, we can see the gap much clearer. Furthermore, we also construct strongly nonlocal sets in specific four-party systems in $\mathbb{C}^3\otimes\mathbb{C}^3\otimes\mathbb{C}^3\otimes\mathbb{C}^3$. These results answer two open questions in \cite{halder2019strong} positively.

Note that in our examples, we can not locally  realize the four-party measurements even by unlimited entangled source. This paper have left some interesting questions for furture works. Is there any general construction of strongly nonlocal OPS in any Hilbert space $\mathbb{C}^{d_1}\otimes\mathbb{C}^{d_1}\otimes\cdots\otimes\mathbb{C}^{d_n}$? Moreover, can we find a smallest strongly nonlocal set in $\mathbb{C}^3\otimes\mathbb{C}^3\otimes\mathbb{C}^3$, more generally in any tripartite systems.

\bibliography{reference}

\begin{thebibliography}{17}%
\makeatletter
\providecommand \@ifxundefined [1]{%
 \@ifx{#1\undefined}
}%
\providecommand \@ifnum [1]{%
 \ifnum #1\expandafter \@firstoftwo
 \else \expandafter \@secondoftwo
 \fi
}%
\providecommand \@ifx [1]{%
 \ifx #1\expandafter \@firstoftwo
 \else \expandafter \@secondoftwo
 \fi
}%
\providecommand \natexlab [1]{#1}%
\providecommand \enquote  [1]{``#1''}%
\providecommand \bibnamefont  [1]{#1}%
\providecommand \bibfnamefont [1]{#1}%
\providecommand \citenamefont [1]{#1}%
\providecommand \href@noop [0]{\@secondoftwo}%
\providecommand \href [0]{\begingroup \@sanitize@url \@href}%
\providecommand \@href[1]{\@@startlink{#1}\@@href}%
\providecommand \@@href[1]{\endgroup#1\@@endlink}%
\providecommand \@sanitize@url [0]{\catcode `\\12\catcode `\$12\catcode
  `\&12\catcode `\#12\catcode `\^12\catcode `\_12\catcode `\%12\relax}%
\providecommand \@@startlink[1]{}%
\providecommand \@@endlink[0]{}%
\providecommand \url  [0]{\begingroup\@sanitize@url \@url }%
\providecommand \@url [1]{\endgroup\@href {#1}{\urlprefix }}%
\providecommand \urlprefix  [0]{URL }%
\providecommand \Eprint [0]{\href }%
\providecommand \doibase [0]{http://dx.doi.org/}%
\providecommand \selectlanguage [0]{\@gobble}%
\providecommand \bibinfo  [0]{\@secondoftwo}%
\providecommand \bibfield  [0]{\@secondoftwo}%
\providecommand \translation [1]{[#1]}%
\providecommand \BibitemOpen [0]{}%
\providecommand \bibitemStop [0]{}%
\providecommand \bibitemNoStop [0]{.\EOS\space}%
\providecommand \EOS [0]{\spacefactor3000\relax}%
\providecommand \BibitemShut  [1]{\csname bibitem#1\endcsname}%
\let\auto@bib@innerbib\@empty
\bibitem [{\citenamefont {Bennett}\ \emph
  {et~al.}(1999{\natexlab{a}})\citenamefont {Bennett}, \citenamefont
  {DiVincenzo}, \citenamefont {Fuchs}, \citenamefont {Mor}, \citenamefont
  {Rains}, \citenamefont {Shor}, \citenamefont {Smolin},\ and\ \citenamefont
  {Wootters}}]{bennett1999quantum}%
  \BibitemOpen
  \bibfield  {author} {\bibinfo {author} {\bibfnamefont {C.~H.}\ \bibnamefont
  {Bennett}}, \bibinfo {author} {\bibfnamefont {D.~P.}\ \bibnamefont
  {DiVincenzo}}, \bibinfo {author} {\bibfnamefont {C.~A.}\ \bibnamefont
  {Fuchs}}, \bibinfo {author} {\bibfnamefont {T.}~\bibnamefont {Mor}}, \bibinfo
  {author} {\bibfnamefont {E.}~\bibnamefont {Rains}}, \bibinfo {author}
  {\bibfnamefont {P.~W.}\ \bibnamefont {Shor}}, \bibinfo {author}
  {\bibfnamefont {J.~A.}\ \bibnamefont {Smolin}}, \ and\ \bibinfo {author}
  {\bibfnamefont {W.~K.}\ \bibnamefont {Wootters}},\ }\href@noop {} {\bibfield
  {journal} {\bibinfo  {journal} {Physical Review A}\ }\textbf {\bibinfo
  {volume} {59}},\ \bibinfo {pages} {1070} (\bibinfo {year}
  {1999}{\natexlab{a}})}\BibitemShut {NoStop}%
\bibitem [{\citenamefont {Bennett}\ \emph
  {et~al.}(1999{\natexlab{b}})\citenamefont {Bennett}, \citenamefont
  {DiVincenzo}, \citenamefont {Mor}, \citenamefont {Shor}, \citenamefont
  {Smolin},\ and\ \citenamefont {Terhal}}]{RN1380}%
  \BibitemOpen
  \bibfield  {author} {\bibinfo {author} {\bibfnamefont {C.~H.}\ \bibnamefont
  {Bennett}}, \bibinfo {author} {\bibfnamefont {D.~P.}\ \bibnamefont
  {DiVincenzo}}, \bibinfo {author} {\bibfnamefont {T.}~\bibnamefont {Mor}},
  \bibinfo {author} {\bibfnamefont {P.~W.}\ \bibnamefont {Shor}}, \bibinfo
  {author} {\bibfnamefont {J.~A.}\ \bibnamefont {Smolin}}, \ and\ \bibinfo
  {author} {\bibfnamefont {B.~M.}\ \bibnamefont {Terhal}},\ }\href {\doibase
  DOI 10.1103/PhysRevLett.82.5385} {\bibfield  {journal} {\bibinfo  {journal}
  {Physical Review Letters}\ }\textbf {\bibinfo {volume} {82}},\ \bibinfo
  {pages} {5385} (\bibinfo {year} {1999}{\natexlab{b}})}\BibitemShut {NoStop}%
\bibitem [{\citenamefont {Walgate}\ and\ \citenamefont {Hardy}(2002)}]{RN1129}%
  \BibitemOpen
  \bibfield  {author} {\bibinfo {author} {\bibfnamefont {J.}~\bibnamefont
  {Walgate}}\ and\ \bibinfo {author} {\bibfnamefont {L.}~\bibnamefont
  {Hardy}},\ }\href {\doibase 10.1103/PhysRevLett.89.147901} {\bibfield
  {journal} {\bibinfo  {journal} {Physical Review Letters}\ }\textbf {\bibinfo
  {volume} {89}},\ \bibinfo {pages} {147901} (\bibinfo {year}
  {2002})}\BibitemShut {NoStop}%
\bibitem [{\citenamefont {Feng}\ and\ \citenamefont {Shi}(2009)}]{RN1377}%
  \BibitemOpen
  \bibfield  {author} {\bibinfo {author} {\bibfnamefont {Y.}~\bibnamefont
  {Feng}}\ and\ \bibinfo {author} {\bibfnamefont {Y.~Y.}\ \bibnamefont {Shi}},\
  }\href {\doibase 10.1109/Tit.2009.2018330} {\bibfield  {journal} {\bibinfo
  {journal} {Ieee Transactions on Information Theory}\ }\textbf {\bibinfo
  {volume} {55}},\ \bibinfo {pages} {2799} (\bibinfo {year}
  {2009})}\BibitemShut {NoStop}%
\bibitem [{\citenamefont {Duan}\ \emph
  {et~al.}(2010{\natexlab{a}})\citenamefont {Duan}, \citenamefont {Xin},\ and\
  \citenamefont {Ying}}]{RN1139}%
  \BibitemOpen
  \bibfield  {author} {\bibinfo {author} {\bibfnamefont {R.}~\bibnamefont
  {Duan}}, \bibinfo {author} {\bibfnamefont {Y.}~\bibnamefont {Xin}}, \ and\
  \bibinfo {author} {\bibfnamefont {M.}~\bibnamefont {Ying}},\ }\href {\doibase
  10.1103/PhysRevA.81.032329} {\bibfield  {journal} {\bibinfo  {journal}
  {Physical Review A}\ }\textbf {\bibinfo {volume} {81}},\ \bibinfo {pages}
  {032329} (\bibinfo {year} {2010}{\natexlab{a}})}\BibitemShut {NoStop}%
\bibitem [{\citenamefont {Childs}\ \emph {et~al.}(2013)\citenamefont {Childs},
  \citenamefont {Leung}, \citenamefont {Mancinska},\ and\ \citenamefont
  {Ozols}}]{RN4268}%
  \BibitemOpen
  \bibfield  {author} {\bibinfo {author} {\bibfnamefont {A.~M.}\ \bibnamefont
  {Childs}}, \bibinfo {author} {\bibfnamefont {D.}~\bibnamefont {Leung}},
  \bibinfo {author} {\bibfnamefont {L.}~\bibnamefont {Mancinska}}, \ and\
  \bibinfo {author} {\bibfnamefont {M.}~\bibnamefont {Ozols}},\ }\href
  {\doibase 10.1007/s00220-013-1784-0} {\bibfield  {journal} {\bibinfo
  {journal} {Communications in Mathematical Physics}\ }\textbf {\bibinfo
  {volume} {323}},\ \bibinfo {pages} {1121} (\bibinfo {year}
  {2013})}\BibitemShut {NoStop}%
\bibitem [{\citenamefont {Yang}\ \emph {et~al.}(2013)\citenamefont {Yang},
  \citenamefont {Gao}, \citenamefont {Tian}, \citenamefont {Cao},\ and\
  \citenamefont {Wen}}]{RN4275}%
  \BibitemOpen
  \bibfield  {author} {\bibinfo {author} {\bibfnamefont {Y.~H.}\ \bibnamefont
  {Yang}}, \bibinfo {author} {\bibfnamefont {F.}~\bibnamefont {Gao}}, \bibinfo
  {author} {\bibfnamefont {G.~J.}\ \bibnamefont {Tian}}, \bibinfo {author}
  {\bibfnamefont {T.~Q.}\ \bibnamefont {Cao}}, \ and\ \bibinfo {author}
  {\bibfnamefont {Q.~Y.}\ \bibnamefont {Wen}},\ }\href {\doibase ARTN 024301
  10.1103/PhysRevA.88.024301} {\bibfield  {journal} {\bibinfo  {journal}
  {Physical Review A}\ }\textbf {\bibinfo {volume} {88}},\ \bibinfo {pages}
  {024301} (\bibinfo {year} {2013})}\BibitemShut {NoStop}%
\bibitem [{\citenamefont {Zhang}\ \emph {et~al.}(2014)\citenamefont {Zhang},
  \citenamefont {Gao}, \citenamefont {Tian}, \citenamefont {Cao},\ and\
  \citenamefont {Wen}}]{RN1133}%
  \BibitemOpen
  \bibfield  {author} {\bibinfo {author} {\bibfnamefont {Z.-C.}\ \bibnamefont
  {Zhang}}, \bibinfo {author} {\bibfnamefont {F.}~\bibnamefont {Gao}}, \bibinfo
  {author} {\bibfnamefont {G.-J.}\ \bibnamefont {Tian}}, \bibinfo {author}
  {\bibfnamefont {T.-Q.}\ \bibnamefont {Cao}}, \ and\ \bibinfo {author}
  {\bibfnamefont {Q.-Y.}\ \bibnamefont {Wen}},\ }\href {\doibase
  10.1103/PhysRevA.90.022313} {\bibfield  {journal} {\bibinfo  {journal}
  {Physical Review A}\ }\textbf {\bibinfo {volume} {90}},\ \bibinfo {pages}
  {022313} (\bibinfo {year} {2014})}\BibitemShut {NoStop}%
\bibitem [{\citenamefont {Yu}(2015)}]{RN1238}%
  \BibitemOpen
  \bibfield  {author} {\bibinfo {author} {\bibfnamefont {S.}~\bibnamefont
  {Yu}},\ }\href@noop {} {\bibfield  {journal} {\bibinfo  {journal} {ArXiv}\ }
  (\bibinfo {year} {2015})}\BibitemShut {NoStop}%
\bibitem [{\citenamefont {Yu}\ \emph {et~al.}(2012)\citenamefont {Yu},
  \citenamefont {Duan},\ and\ \citenamefont {Ying}}]{RN567}%
  \BibitemOpen
  \bibfield  {author} {\bibinfo {author} {\bibfnamefont {N.}~\bibnamefont
  {Yu}}, \bibinfo {author} {\bibfnamefont {R.}~\bibnamefont {Duan}}, \ and\
  \bibinfo {author} {\bibfnamefont {M.}~\bibnamefont {Ying}},\ }\href {\doibase
  10.1103/PhysRevLett.109.020506} {\bibfield  {journal} {\bibinfo  {journal}
  {Phys Rev Lett}\ }\textbf {\bibinfo {volume} {109}},\ \bibinfo {pages}
  {020506} (\bibinfo {year} {2012})}\BibitemShut {NoStop}%
\bibitem [{\citenamefont {Nathanson}(2013)}]{RN1157}%
  \BibitemOpen
  \bibfield  {author} {\bibinfo {author} {\bibfnamefont {M.}~\bibnamefont
  {Nathanson}},\ }\href {\doibase 10.1103/PhysRevA.88.062316} {\bibfield
  {journal} {\bibinfo  {journal} {Physical Review A}\ }\textbf {\bibinfo
  {volume} {88}},\ \bibinfo {pages} {062316} (\bibinfo {year}
  {2013})}\BibitemShut {NoStop}%
\bibitem [{\citenamefont {Chitambar}\ \emph {et~al.}(2013)\citenamefont
  {Chitambar}, \citenamefont {Duan},\ and\ \citenamefont
  {Hsieh}}]{chitambar2013local}%
  \BibitemOpen
  \bibfield  {author} {\bibinfo {author} {\bibfnamefont {E.}~\bibnamefont
  {Chitambar}}, \bibinfo {author} {\bibfnamefont {R.}~\bibnamefont {Duan}}, \
  and\ \bibinfo {author} {\bibfnamefont {M.-H.}\ \bibnamefont {Hsieh}},\
  }\href@noop {} {\bibfield  {journal} {\bibinfo  {journal} {IEEE Transactions
  on Information Theory}\ }\textbf {\bibinfo {volume} {60}},\ \bibinfo {pages}
  {1549} (\bibinfo {year} {2013})}\BibitemShut {NoStop}%
\bibitem [{\citenamefont {Tian}\ \emph {et~al.}(2015)\citenamefont {Tian},
  \citenamefont {Yu}, \citenamefont {Gao}, \citenamefont {Wen},\ and\
  \citenamefont {Oh}}]{RN1966}%
  \BibitemOpen
  \bibfield  {author} {\bibinfo {author} {\bibfnamefont {G.~J.}\ \bibnamefont
  {Tian}}, \bibinfo {author} {\bibfnamefont {S.~X.}\ \bibnamefont {Yu}},
  \bibinfo {author} {\bibfnamefont {F.}~\bibnamefont {Gao}}, \bibinfo {author}
  {\bibfnamefont {Q.~Y.}\ \bibnamefont {Wen}}, \ and\ \bibinfo {author}
  {\bibfnamefont {C.~H.}\ \bibnamefont {Oh}},\ }\href {\doibase ARTN 042320
  10.1103/PhysRevA.92.042320} {\bibfield  {journal} {\bibinfo  {journal}
  {Physical Review A}\ }\textbf {\bibinfo {volume} {92}} (\bibinfo {year}
  {2015}),\ ARTN 042320 10.1103/PhysRevA.92.042320}\BibitemShut {NoStop}%
\bibitem [{\citenamefont {Liu}\ \emph {et~al.}(2019)\citenamefont {Liu},
  \citenamefont {Li},\ and\ \citenamefont {Duan}}]{liu2019distinguishing}%
  \BibitemOpen
  \bibfield  {author} {\bibinfo {author} {\bibfnamefont {S.}~\bibnamefont
  {Liu}}, \bibinfo {author} {\bibfnamefont {Y.}~\bibnamefont {Li}}, \ and\
  \bibinfo {author} {\bibfnamefont {R.}~\bibnamefont {Duan}},\ }\href@noop {}
  {\bibfield  {journal} {\bibinfo  {journal} {Science China Information
  Sciences}\ }\textbf {\bibinfo {volume} {62}},\ \bibinfo {pages} {72502}
  (\bibinfo {year} {2019})}\BibitemShut {NoStop}%
\bibitem [{\citenamefont {Duan}\ \emph
  {et~al.}(2010{\natexlab{b}})\citenamefont {Duan}, \citenamefont {Xin},\ and\
  \citenamefont {Ying}}]{duan2010locally}%
  \BibitemOpen
  \bibfield  {author} {\bibinfo {author} {\bibfnamefont {R.}~\bibnamefont
  {Duan}}, \bibinfo {author} {\bibfnamefont {Y.}~\bibnamefont {Xin}}, \ and\
  \bibinfo {author} {\bibfnamefont {M.}~\bibnamefont {Ying}},\ }\href@noop {}
  {\bibfield  {journal} {\bibinfo  {journal} {Physical Review A}\ }\textbf
  {\bibinfo {volume} {81}},\ \bibinfo {pages} {032329} (\bibinfo {year}
  {2010}{\natexlab{b}})}\BibitemShut {NoStop}%
\bibitem [{\citenamefont {Halder}\ \emph {et~al.}(2019)\citenamefont {Halder},
  \citenamefont {Banik}, \citenamefont {Agrawal},\ and\ \citenamefont
  {Bandyopadhyay}}]{halder2019strong}%
  \BibitemOpen
  \bibfield  {author} {\bibinfo {author} {\bibfnamefont {S.}~\bibnamefont
  {Halder}}, \bibinfo {author} {\bibfnamefont {M.}~\bibnamefont {Banik}},
  \bibinfo {author} {\bibfnamefont {S.}~\bibnamefont {Agrawal}}, \ and\
  \bibinfo {author} {\bibfnamefont {S.}~\bibnamefont {Bandyopadhyay}},\
  }\href@noop {} {\bibfield  {journal} {\bibinfo  {journal} {Physical review
  letters}\ }\textbf {\bibinfo {volume} {122}},\ \bibinfo {pages} {040403}
  (\bibinfo {year} {2019})}\BibitemShut {NoStop}%
\bibitem [{\citenamefont {Zhang}\ and\ \citenamefont
  {Zhang}(2019)}]{zhang2019strong}%
  \BibitemOpen
  \bibfield  {author} {\bibinfo {author} {\bibfnamefont {Z.-C.}\ \bibnamefont
  {Zhang}}\ and\ \bibinfo {author} {\bibfnamefont {X.}~\bibnamefont {Zhang}},\
  }\href@noop {} {\bibfield  {journal} {\bibinfo  {journal} {Physical Review
  A}\ }\textbf {\bibinfo {volume} {99}},\ \bibinfo {pages} {062108} (\bibinfo
  {year} {2019})}\BibitemShut {NoStop}%
\end{thebibliography}%

\appendix
\section{Proof of Lemma \ref{lem:3333_irreducible2}}
\label{app:lem3333irreducible2}
Since under the cyclic permutation of the parties, the bipartition $AB|CD$ is the same as bipartion $AD|BC$, we only discuss the cases $AB|CD$ and $AC|BD$.
Here, we only consider a subset of (\ref{3333_nonlocal}): (\ref{3333_irreducible}) and (\ref{3333_irreducible2}).

For bipartition $AB|CD$, let POVM $\{\Pi_{AB}\}$ denote a general orhthogonal-preserving measurement on composite space $AB$:
\[
\Pi_{AB}=\left(
           \begin{array}{ccc}
             a_{00,00} & \cdots & a_{00,22} \\
             \vdots & \ddots  & \vdots \\
             a_{22,00} & \cdots  & a_{22,22} \\
           \end{array}
         \right)\in\mathbb{C}^{9\times 9}.
\]

Since the measurement leaves the postmeasurement states mutually orthogonal, all off-diagonal elements of $\Pi_{AB}$ are equal to zero. See details in Table \ref{tab:lem_3333_irreducible2}.
\begin{table}[!htb]
    \centering
    \begin{tabular}{|c|c|c|}
    \hline
       No.  & States & Zero-Elements \\
    \hline
    \hline
    \multirow {2}{*}{1} & \multirow {2}{*}{$|\phi_{1}\rangle,|\phi_{29,30,31,32}\rangle$} & $a_{01,11},a_{01,10},a_{01,21},a_{01,20}$\\
      &  & $a_{11,01},a_{10,01},a_{21,01},a_{20,01}$ \\
     \hline
     \multirow{2}{*}{2 } & \multirow{2}{*}{$|\phi_{9}\rangle,|\phi_{25,26,27,28}\rangle$} & $a_{20,00},a_{20,02},a_{20,10},a_{20,12}$
     \\
      &  & $a_{00,20},a_{02,20},a_{10,20},a_{12,20}$ \\
      \hline
     \multirow {4}{*}{3} & \multirow {4}{*}{$|\phi_{17,19}\rangle,|\phi_{25,26,27,28}\rangle$} & $a_{21,00},a_{21,02},a_{21,10},a_{21,12}$\\
      &  & $a_{22,00},a_{22,02},a_{22,10},a_{22,12}$ \\
      &  & $a_{00,21},a_{02,21},a_{10,21},a_{12,21}$\\
      &  & $a_{00,22},a_{02,22},a_{10,22},a_{12,22}$\\
     \hline
      \multirow {2}{*}{4} & \multirow {2}{*}{$|\psi_{1}\rangle,|\phi_{25,26,27,28}\rangle$} & $a_{11,00},a_{11,02},a_{11,10},a_{11,12}$\\
      &  & $a_{00,11},a_{02,11},a_{10,11},a_{12,11}$\\
      \hline
     5 & $|\psi_{7,8}\rangle,|\phi_{37,39}\rangle$ & $a_{01,00},a_{02,00},a_{00,01},a_{00,02}$\\
     \hline
     6 & $|\psi_{9,10}\rangle,|\phi_{21,23}\rangle$ & $a_{10,00},a_{12,00},a_{00,10},a_{00,12}$\\
     \hline
     7 & $|\psi_{7,8}\rangle,|\phi_{5}\rangle$ & $a_{12,01},a_{12,02},a_{01,12},a_{02,12}$\\
     \hline
    8  & $|\psi_{11,12}\rangle,|\phi_{45,47}\rangle$ & $a_{22,20},a_{22,21},a_{20,22},a_{20,21}$\\
     \hline
    9 & $|\psi_{1,2}\rangle,|\phi_{17,19}\rangle$ & $a_{11,21},a_{11,22},a_{21,11},a_{22,11}$ \\
     \hline
     10 & $|\psi_{21,22}\rangle,|\phi_{41,43}\rangle$ & $a_{01,02},a_{01,22},a_{02,01},a_{22,01}$\\
     \hline
    11 & $|\psi_{19,20}\rangle,|\phi_{37,39}\rangle$ & $a_{20,11},a_{20,21},a_{11,20},a_{21,20}$\\
     \hline
   12  & $|\psi_{13}\rangle,|\psi_{7,8}\rangle$ & $a_{10,02},a_{02,10}$\\
     \hline
    13 & $|\psi_{13}\rangle,|\phi_{5}\rangle$ & $a_{10,12},a_{12,10}$\\
     \hline
    \end{tabular}
    \caption{Off-diagonal elements of $\Pi_{AB}$}
    \label{tab:lem_3333_irreducible2}
\end{table}
By setting $\langle \phi_{25}|\Pi_{AB}\otimes\mathbb{I}\otimes\mathbb{I}|\phi_{26,27,28}\rangle=\langle\phi_{29}|\Pi_{AB}\otimes\mathbb{I}\otimes\mathbb{I}|\phi_{30,31,32}\rangle=\langle\phi_{33}|\Pi_{AB}\otimes\mathbb{I}\otimes\mathbb{I}|\phi_{34,35,36}\rangle=0$, we obtain $a_{00,00}=a_{01,01}=\cdots=a_{22,22}$. Since $\Pi_{AB}$ is proportional to an identity operator, $AB$ can not go first. Because composite space $AB$ is the same as composite space $CD$, $CD$ can not go first. Hence, in bipartition $AB|CD$, (\ref{3333_nonlocal}) is locally irreducible. Quantum states (\ref{3333_irreducible}) and (\ref{3333_irreducible2}) are the same in bipartition $AD|BC$ and $AB|CD$, therefore (\ref{3333_nonlocal}) is locally irreducible in bipartition $AD|BC$.

For bipartition $AC|BD$,let POVM $\{\Pi_{AC}\}$ denote a general orthogonal-preserving measurement on composite space $AC$:
\[
\Pi_{AC}=\left(
           \begin{array}{ccc}
             b_{00,00} & \cdots & b_{00,22} \\
             \vdots & \ddots  & \vdots \\
             b_{22,00} & \cdots  & b_{22,22} \\
           \end{array}
         \right)\in\mathbb{C}^{9\times 9}.
\]

As the post measurement states leave pairwise orthogonal to each other, all off-diagonal elements of $\Pi_{AC}$ are all equal to zeros. Table \ref{tab:lem_3333_irreducible22} shows the details.
\begin{table}[!htb]
    \centering
    \begin{tabular}{|c|c|c|}
    \hline
       No.  & States & Zero-Elements \\
    \hline
    \hline
      1 & $|\psi_{1}\rangle,|\phi_{41,43}\rangle$ & $b_{10,12},b_{10,02},b_{12,10},b_{02,10}$ \\
    \hline
     2 & $|\psi_{5}\rangle,|\phi_{37,39}\rangle$ & $b_{02,01},b_{02,21},b_{01,02},b_{21,02}$ \\
     \hline
       3  & \multirow{2}{*}{$|\psi_{23,24}\rangle,|\phi_{9,11}\rangle$}  & $b_{00,20},b_{00,22},b_{10,20},b_{10,22}$\\
          & & $b_{20,00},b_{22,00},b_{20,10},b_{22,10}$\\
     \hline
      \multirow{2}{*}{4}  & \multirow{2}{*}{$|\phi_{1,3}\rangle,|\phi_{29,31}\rangle$}  &  $b_{00,11},b_{00,21},b_{01,11},b_{01,21}$\\
        & & $b_{11,00},b_{21,00},b_{11,01},b_{21,01}$\\
    \hline
      \multirow{2}{*}{5}  & \multirow{2}{*}{$|\phi_{5,7}\rangle,|\phi_{33,35}\rangle$} & $b_{11,22},b_{11,02},b_{12,22},b_{12,02}$ \\
        & & $b_{22,11},b_{02,11},b_{22,12},b_{02,1}$ \\
    \hline
     \multirow{2}{*}{6}  & \multirow{2}{*}{$|\phi_{9,11}\rangle,|\phi_{37,39}\rangle$} & $b_{22,01},b_{22,21},b_{20,01},b_{20,21}$ \\
        & & $b_{01,22},b_{21,22},b_{01,20},b_{21,20}$ \\
    \hline
     \multirow{2}{*}{7}  & \multirow{2}{*}{$|\phi_{13,14}\rangle,|\phi_{37,39}\rangle$} & $b_{10,01},b_{10,21},b_{12,01},b_{12,21}$\\
       & & $b_{01,10},b_{21,10},b_{01,12},b_{21,12}$\\
     \hline
    8 & $|\phi_{1,3}\rangle,|\phi_{13,14}\rangle$ & $b_{00,10},b_{00,12},b_{10,00},b_{12,00}$\\
     \hline
   \multirow{2}{*}{9}  & \multirow{2}{*}{$|\phi_{5,7}\rangle,|\phi_{17,18}\rangle$} & $b_{11,20},b_{11,21},b_{12,20},b_{12,21}$\\
     & & $b_{20,11},b_{21,11},b_{20,12},b_{21,12}$ \\
     \hline
   10  & $|\phi_{9,11}\rangle,|\phi_{21,22}\rangle$ & $b_{22,02},b_{20,02},b_{02,22},b_{02,20}$ \\
     \hline
   11  & $|\psi_{23,24}\rangle,|\phi_{21,22}\rangle$ & $b_{00,01},b_{00,02},b_{01,00},b_{01,00}$ \\
     \hline
   12  & $|\psi_{19,20}\rangle,|\phi_{13,14}\rangle$ & $b_{11,10},b_{11,12},b_{10,11},b_{12,11}$\\
     \hline
13 & $|\psi_{11}\rangle,|\phi_{33,35}\rangle$ & $b_{20,22},b_{22,20}$\\
     \hline
    \end{tabular}
    \caption{Off-diagonal elements of $\Pi_{AC}$}
    \label{tab:lem_3333_irreducible22}
\end{table}
Since the post measurement states preserve orthogonality, all diagonal elements of $\Pi_{AC}$ are equal, see in Table \ref{tab:lem_3333_irreducible222}.
\begin{table}[!htb]
    \centering
    \begin{tabular}{|c|c|c|}
    \hline
       No.  & States & Elements\\
    \hline
    \hline
       1  & $|\psi_{19,20}\rangle$ & $b_{11,11}=b_{21,21}$ \\
    \hline
       2  & $|\psi_{21,22}\rangle$ & $b_{22,22}=b_{02,02}$ \\
    \hline
       3  & $|\phi_{1,3}\rangle$ & $b_{00,00}=b_{01,01}$\\
    \hline
       4  & $|\phi_{5,7}\rangle$ & $b_{11,11}=b_{12,12}$ \\
     \hline
     5 & $|\phi_{9,11}\rangle$ & $b_{22,22}=b_{20,20}$ \\
     \hline
     6 & $|\phi_{13,14}\rangle$ & $b_{10,10}=b_{12,12}$ \\
     \hline
     7 & $|\phi_{17,18}\rangle$ & $b_{21,21}=b_{20,20}$ \\
     \hline
     8 & $|\phi_{21,22}\rangle$ &$b_{02,02}=b_{01,01}$ \\
    \hline
    \end{tabular}
    \caption{Diagonal elements of $\Pi_{AC}$}
    \label{tab:lem_3333_irreducible222}
\end{table}

 Sine $\Pi_{AC}$ is proportional to an identity operator, $AC$ can not go first. Since composite space $AC$ is the same as composite space $BD$, $BD$ can not go first. Hence, (\ref{3333_nonlocal}) is locally irreducible in bipartition $AC|BD$. This completes the proof of Lemma \ref{lem:3333_irreducible2}.

\section{Proof of Lemma \ref{lem:3333_irreducible3}}
\label{app:lem3333irreducible3}
In this section, we will complete the proof of Lemma \ref{lem:3333_irreducible3}. Firstly, we will show that is locally irreducible in bipartition $A|BCD$. Due to Lemma \ref{lem:3333_irreducible}, $A$ can not go first. We now assume that $BCD$ goes first. Let the POVM $\{\Pi_{BCD}\}$ describe a general orthogonality-preserving measurement on $BCD$. Each POVM element $\Pi_{BCD}$ can be written as $27\times 27$ matrix in the basis $\{|000\rangle,|001\rangle,\cdots,|221\rangle,|222\rangle\}$ basis of $\mathcal{H}_{BCD}$:
\[
\Pi_{BCD}=\left(
           \begin{array}{ccc}
             c_{000,000} & \cdots & c_{000,222} \\
             \vdots & \ddots  & \vdots \\
             c_{222,000} & \cdots  & c_{222,222} \\
           \end{array}
         \right)\in\mathbb{C}^{27\times 27}.
\]

Since every measurement leaves postquantum states mutually orthogonal, all off-diagonal elements of $\Pi_{BCD}$ are equal to zero. The details are presented in Table \ref{tab:3333_irreducible3}. In this table, we only show half of off-diagonal elements of $\Pi_{BCD}$, the rest of off-diagonal elements are equal to zero (i.e., $a_{i,j}=a_{ji}$ for $i,j\in\{000,\cdots,222\}$) since $\Pi_{BCD}$ is positive definite.
\begin{table*}[!htb]
    \centering
    \begin{tabular}{|c|c|c||c|c|c|}
    \hline
      No.  & States  & Zero-Elements & No.  & States  & Zero-Elements\\
    \hline
    \hline
        \multirow {4}{*}{1} & & $c_{100,212},c_{100,222},c_{100,022},c_{100,012}$ & \multirow {2}{*}{33} & $|\phi_{21,22,23,24}\rangle$ & $c_{001,222},c_{001,212},c_{001,022},c_{001,012}$\\
         & $|\phi_{1,2,3,4}\rangle$ & $c_{102,212},c_{102,222},c_{102,022},c_{102,012}$& & $|\phi_{25,26,27,28}\rangle$ & $c_{201,222},c_{201,212},c_{201,022},c_{201,012}$\\
    \cline{4-6}
         & $|\phi_{21,22,23,24}\rangle$ & $c_{110,212},c_{110,222},c_{110,022},c_{110,012}$ & \multirow {2}{*}{34} & $|\phi_{21,22,23,24}\rangle$ & $c_{220,222},c_{220,212},c_{220,022},c_{220,012}$\\
         &  & $c_{112,212},c_{112,222},c_{112,022},c_{112,012}$ & & $|\phi_{33,34,35,36}\rangle$ & $c_{120,222},c_{120,212},c_{120,022},c_{120,012}$\\
    \hline
         \multirow {2}{*}{2} & $|\phi_{1,2,3,4}\rangle$ & $c_{001,100}, c_{001,102},c_{001,110},c_{001,112}$ & \multirow {2}{*}{35} & $|\phi_{21,22,23,24}\rangle$ & $c_{010,222},c_{010,212},c_{010,022},c_{010,012}$\\
         & $|\phi_{25,26,27,28}\rangle$ & $c_{201,100}, c_{201,102},c_{201,110},c_{201,112}$ & & $|\phi_{37,38,39,40}\rangle$ & $c_{011,222},c_{011,212},c_{011,022},c_{011,012}$\\
    \hline
         \multirow {2}{*}{3}& $|\phi_{1,2,3,4}\rangle$ & $c_{220,100}, c_{220,102},c_{220,110},c_{220,112}$ & \multirow {2}{*}{36} & $|\phi_{21,22,23,24}\rangle$ & $c_{121,222},c_{121,212},c_{121,022},c_{121,012}$\\
         & $|\phi_{33,34,35,36}\rangle$ & $c_{120,100}, c_{120,102},c_{120,110},c_{120,112}$ & & $|\phi_{41,42,43,44}\rangle$ & $c_{122,222},c_{122,212},c_{122,022},c_{122,012}$\\
    \hline
     \multirow {2}{*}{4} & $|\phi_{1,2,3,4}\rangle$ & $c_{010,100},c_{010,102},c_{010,110},c_{010,112}$ & \multirow {2}{*}{37} & $|\psi_{1,2}\rangle$ & $c_{101,211},c_{101,210},c_{101,221},c_{101,220}$\\
     & $|\phi_{37,38,39,40}\rangle$ & $c_{011,100},c_{011,102},c_{011,110},c_{011,112}$ & & $|\phi_{5,6,7,8}\rangle$ & $c_{102,211},c_{102,210},c_{102,221},c_{102,220}$\\
     \hline
     \multirow {2}{*}{5} & $|\phi_{1,2,3,4}\rangle$ & $c_{121,100},c_{121,102},c_{121,110},c_{121,112}$ & \multirow {2}{*}{38} & $|\psi_{3,4}\rangle$ & $c_{212,022},c_{212,021},c_{212,002},c_{212,001}$\\
     & $|\phi_{41,42,43,44}\rangle$ & $c_{122,100},c_{122,102},c_{122,110},c_{122,112}$ & & $|\phi_{9,10,11,12}\rangle$ & $c_{210,022},c_{210,021},c_{210,002},c_{210,001}$\\
     \hline
     \multirow {4}{*}{6}& & $c_{211,000},c_{211,020},c_{211,100},c_{211,120}$ & \multirow {2}{*}{39} & $|\psi_{5,6}\rangle$ & $c_{020,100},c_{020,102},c_{020,110},c_{020,112}$\\
     & $|\phi_{5,6,7,8}\rangle$ & $c_{210,000},c_{210,020},c_{210,100},c_{210,120}$ & & $|\phi_{1,2,3,4}\rangle$ & $c_{021,100},c_{021,102},c_{021,110},c_{021,112}$\\
     \cline{4-6}
     & $|\phi_{13,14,15,16}\rangle$ & $c_{221,000},c_{221,020},c_{221,100},c_{221,120}$ & \multirow {2}{*}{40} & $|\psi_{1,2}\rangle$ & $c_{101,000},c_{101,020},c_{101,100},c_{101,120}$\\
     &  & $c_{220,000},c_{220,020},c_{220,100},c_{220,120}$ & & $|\phi_{13,14,15,16}\rangle$ & $c_{102,000},c_{102,020},c_{102,100},c_{102,120}$\\
     \hline
     \multirow {2}{*}{7}& $|\phi_{5,6,7,8}\rangle$ & $c_{001,211},c_{001,210},c_{001,221},c_{001,220}$ & \multirow {2}{*}{41} & $|\psi_{3,4}\rangle$ & $c_{212,111},c_{212,101},c_{212,211},c_{212,201}$\\
     & $|\phi_{25,26,27,28}\rangle$ & $c_{201,211},c_{201,210},c_{201,221},c_{201,220}$ & & $|\phi_{17,18,19,20}\rangle$ & $c_{210,111},c_{210,101},c_{210,211},c_{210,201}$ \\
     \hline
     \multirow {2}{*}{8} & $|\phi_{5,6,7,8}\rangle$ & $c_{112,211},c_{112,210},c_{112,221},c_{112,220}$ & \multirow {2}{*}{42} & $|\psi_{7,8}\rangle$ & $c_{111,100},c_{111,102},c_{111,110},c_{111,112}$\\
     & $|\phi_{29,30,31,32}\rangle$ & $c_{012,211},c_{012,210},c_{012,221},c_{012,220}$ & & $|\phi_{1,2,3,4}\rangle$ & $c_{211,100},c_{211,102},c_{211,110},c_{211,112}$\\
     \hline
     \multirow {2}{*}{9} & $|\phi_{5,6,7,8}\rangle$ & $c_{121,211},c_{121,210},c_{121,221},c_{121,220}$ & \multirow {2}{*}{43} &$|\psi_{9,10}\rangle$ & $c_{222,211},c_{222,210},c_{222,221},c_{222,220}$\\
     & $|\phi_{41,42,43,44}\rangle$ & $c_{122,211},c_{122,210},c_{122,221},c_{122,220}$ & & $|\phi_{5,6,7,8}\rangle$ & $c_{022,211},c_{022,210},c_{022,221},c_{022,220}$\\
     \hline
     \multirow {2}{*}{10}& $|\phi_{5,6,7,8}\rangle$ & $c_{202,211},c_{202,210},c_{202,221},c_{202,220}$ & \multirow {2}{*}{44} & $|\psi_{11,12}\rangle$ & $c_{000,022},c_{000,021},c_{000,002},c_{000,001}$ \\
     & $|\phi_{45,46,47,48}\rangle$ & $c_{200,211},c_{200,210},c_{200,221},c_{200,220}$ & & $|\phi_{9,10,11,12}\rangle$ & $c_{100,022},c_{100,021},c_{100,002},c_{100,001}$\\
     \hline
     \multirow {4}{*}{11} & & $c_{111,022},c_{111,021},c_{111,002},c_{111,001}$ & 45 & $|\psi_{19}\rangle,|\phi_{5,6,7,8}\rangle$ & $c_{110,211},c_{110,210},c_{110,221},c_{110,220}$\\
     \cline{4-6}
     & $|\phi_{9,10,11,12}\rangle$ & $c_{101,022},c_{101,021},c_{101,002},c_{101,001}$ &46 & $|\psi_{19}\rangle,|\phi_{13,14,15,16}\rangle$ & $c_{110,000},c_{110,020},c_{110,100},c_{110,120}$\\
     \cline{4-6}
     & $|\phi_{17,18,19,20}\rangle$ & $c_{211,022},c_{211,021},c_{211,002},c_{211,001}$ &47 & $|\psi_{21}\rangle,|\phi_{9,10,11,12}\rangle$ & $c_{221,002},c_{221,001},c_{221,022},c_{221,021}$\\
     \cline{4-6}
     & & $c_{201,022},c_{201,021},c_{201,002},c_{201,001}$ & 48 & $|\psi_{21}\rangle,|\phi_{17,18,19,20}\rangle$ & $c_{221,111},c_{221,101},c_{221,211},c_{221,201}$\\
     \hline
     \multirow {2}{*}{12} & $|\phi_{9,10,11,12}\rangle$ & $c_{112,022},c_{112,021},c_{112,002},c_{112,001}$ & 49 & $|\psi_{23}\rangle,|\phi_{1,2,3,4}\rangle$ & $c_{002,110},c_{002,112},c_{002,100},c_{002,102}$ \\
     \cline{4-6}
     & $|\phi_{29,30,31,32}\rangle$ & $c_{012,022},c_{012,021},c_{012,002},c_{012,001}$& 50 & $|\psi_{23}\rangle,|\phi_{21,22,23,24}\rangle$ & $c_{002,222},c_{002,212},c_{002,022},c_{002,012}$\\
     \hline
     \multirow {2}{*}{13} & $|\phi_{9,10,11,12}\rangle$ & $c_{220,022},c_{220,021},c_{220,002},c_{220,001}$ & 51 & $|\varphi_1\rangle,|\phi_{41,42,45,46}\rangle$ & $c_{010,121},c_{010,122},c_{010,202},c_{010,200}$\\
     \cline{4-6}
      & $|\phi_{33,34,35,36}\rangle$ & $c_{120,022},c_{120,021},c_{120,002},c_{120,001}$ & 52 & $|\varphi_2\rangle,|\phi_{37,38,45,46}\rangle$ & $c_{121,010},c_{121,011},c_{121,202},c_{121,200}$\\
     \hline
     \multirow {2}{*}{14}& $|\phi_{9,10,11,12}\rangle$ & $c_{010,022},c_{010,021},c_{010,002},c_{010,001}$ & 53 & $|\varphi_{3}\rangle,|\phi_{37,38,41,42}\rangle$ & $c_{202,010},c_{202,011},c_{202,121},c_{202,122}$\\
     \cline{4-6}
     & $|\phi_{37,28,29,40}\rangle$ & $c_{011,022},c_{011,021},c_{011,002},c_{011,001}$ & 54 & $|\phi_{25,26}\rangle,|\phi_{41,42}\rangle$ & $c_{001,121},c_{001,122},c_{201,122},c_{201,121}$\\
     \hline
     \multirow {2}{*}{15} & $|\phi_{9,10,11,12}\rangle$ & $c_{201,022},c_{201,021},c_{201,002},c_{201,001}$& 55 & $|\phi_{29,30}\rangle,|\phi_{45,46}\rangle$ & $c_{112,202},c_{112,200},c_{012,202},c_{012,200}$\\
     \cline{4-6}
     & $|\phi_{45,46,47,48}\rangle$ & $c_{20,022},c_{200,021},c_{200,002},c_{200,001}$ & 56 & $|\phi_{33,34}\rangle,|\phi_{37,38}\rangle$ & $c_{220,010},c_{220,011},c_{120,010},c_{120,011}$\\
     \hline
     \multirow {2}{*}{16} & $|\phi_{13,14,15,16}\rangle$ & $c_{001,000},c_{001,020},c_{001,100},c_{001,120}$ &57 & $|\varphi_{5}\rangle,|\phi_{29,30,45,46}\rangle$ & $c_{212,112},c_{212,012},c_{212,202},c_{212,200}$\\
     \cline{4-6}
     & $|\phi_{25,26,27,28}\rangle$ & $c_{201,000},c_{201,020},c_{201,100},c_{201,120}$ & 58 & $|\varphi_{6}\rangle,|\phi_{33,34,37,38}\rangle$ & $c_{020,010},c_{020,011},c_{020,220},c_{020,120}$\\
     \hline
     \multirow {2}{*}{17} & $|\phi_{13,14,15,16}\rangle$ & $c_{112,000},c_{112,020},c_{112,100},c_{112,120}$& 59 & $|\varphi_{4}\rangle,|\phi_{25,26,41,42}\rangle$ & $c_{101,121},c_{101,122},c_{101,001},c_{101,201}$\\
     \cline{4-6}
     & $|\phi_{29,30,31,32}\rangle$ & $c_{012,000},c_{012,020},c_{012,100},c_{012,120}$ & 60 & $|\varphi_{7}\rangle,|\phi_{21,22,23,24}\rangle$ & $c_{000,222},c_{000,212},c_{000,022},c_{000,012}$\\
     \hline
     \multirow {2}{*}{18} & $|\phi_{13,14,15,16}\rangle$ & $c_{121,000},c_{121,020},c_{121,100},c_{121,120}$  & 61 & $|\varphi_{8}\rangle,|\phi_{13,14,15,16}\rangle$ & $c_{111,000},c_{111,020},c_{111,100},c_{111,120}$\\
     \cline{4-6}
     & $|\phi_{41,42,43,44}\rangle$ & $c_{122,000},c_{122,020},c_{122,100},c_{122,120}$ & 62 & $|\varphi_{9}\rangle,|\phi_{9,10,11,12}\rangle$ & $c_{222,111},c_{222,101},c_{222,211},c_{222,201}$\\
     \hline
     \multirow {2}{*}{19} & $|\phi_{13,14,15,16}\rangle$ & $c_{202,000},c_{202,020},c_{202,100},c_{202,120}$ & 63 & $|\psi_{23}\rangle,|\psi_{5,6},\phi_{41,42}\rangle$ & $c_{002,020},c_{002,021},c_{002,121},c_{002,122}$\\
     \cline{4-6}
     & $|\phi_{45,46,47,48}\rangle$ & $c_{200,000},c_{200,020},c_{200,100},c_{200,120}$ & 64 & $|\psi_{21}\rangle,|\psi_{3,4},\phi_{37,38}\rangle$  & $c_{221,212},c_{221,210},c_{221,010},c_{221,011}$\\
     \hline
     \multirow {2}{*}{20} & $|\phi_{17,18,19,20}\rangle$ & $c_{112,111},c_{112,101},c_{112,211},c_{112,201}$& 65& $|\psi_{19}\rangle,|\psi_{1,2},\phi_{45,46}\rangle$  & $c_{110,101},c_{110,102},c_{110,202},c_{110,200}$ \\
     \cline{4-6}
     & $|\phi_{29,30,31,32}\rangle$ & $c_{012,111},c_{012,101},c_{012,211},c_{012,201}$ & 66 & $|\phi_{41,42}\rangle,|\psi_{13,14}\rangle$ & $c_{121,011},c_{121,021},c_{122,011},c_{122,021}$\\
     \hline
     \multirow {2}{*}{21} & $|\phi_{17,18,19,20}\rangle$ & $c_{220,111},c_{220,101},c_{220,211},c_{220,201}$  &67 & $|\phi_{45,46}\rangle,|\psi_{15,16}\rangle$ & $c_{202,122},c_{202,102},c_{200,122},c_{200,102}$ \\
     \cline{4-6}
     & $|\phi_{33,34,35,36}\rangle$ & $c_{120,111},c_{120,101},c_{120,211},c_{120,201}$   & 68 & $|\phi_{37,38}\rangle,|\psi_{17,18}\rangle$ & $c_{010,200},c_{010,210},c_{011,200},c_{011,210}$  \\
     \hline
     \multirow {2}{*}{22} & $|\phi_{17,18,19,20}\rangle$ & $c_{010,111},c_{010,101},c_{010,211},c_{010,201}$ & 69& $|\varphi_{7}\rangle,|\phi_{37,38,33,34}\rangle$ & $c_{000,010},c_{000,011},c_{000,220},c_{000,120}$\\
     \cline{4-6}
     & $|\phi_{37,38,39,40}\rangle$ & $c_{011,111},c_{011,101},c_{011,211},c_{011,201}$ & 70 & $|\varphi_{8}\rangle,|\phi_{25,26,41,42}\rangle$ & $c_{111,121},c_{111,122},c_{111,001},c_{111,201}$ \\
     \hline
     \multirow {2}{*}{23} & $|\phi_{17,18,19,20}\rangle$ & $c_{202,111},c_{202,101},c_{202,211},c_{202,201}$  & 71 & $|\varphi_{9}\rangle,|\phi_{29,30,45,46}\rangle$ & $c_{222,202},c_{222,200},c_{222,112},c_{222,012}$  \\
     \cline{4-6}
      & $|\phi_{45,46,47,48}\rangle$ & $c_{200,111},c_{200,101},c_{200,211},c_{200,201}$ & 72 & $|\psi_{3,4}\rangle,|\phi_{33,34}\rangle$ & $c_{212,220},c_{212,120},c_{210,220},c_{210,120}$ \\
     \hline
     24 & $|\psi_{5,6}\rangle,|\phi_{25,26}\rangle$ & $c_{020,001},c_{020,201},c_{021,001},c_{021,201}$ & 73 & $|\psi_{1,2}\rangle,|\phi_{29,30}\rangle$ & $c_{101,112},c_{101,012},c_{102,112},c_{102,012}$\\
     \hline
     25 & $|\varphi_{7}\rangle,|\psi_{5,6}\rangle$ & $c_{000,020},c_{000,021}$ &74 & $|\varphi_{7}\rangle,|\phi_{1,2,3,4}\rangle$ & $c_{000,100},c_{000,102},c_{000,110},c_{000,112}$\\
     \hline
     26 & $|\psi_{23}\rangle,|\phi_{25,26}\rangle$ & $c_{002,002},c_{002,201}$ &75 & $|\varphi_{1}\rangle,|\psi_{13,14}\rangle$ & $c_{010,011},c_{010,021}$\\
     \hline
     27 & $|\varphi_{6}\rangle,|\phi_{9,10,11,12}\rangle$ & $c_{020,002},c_{020,001},c_{020,022},c_{020,021}$ & 76 & $|\varphi_{9}\rangle,|\phi_{9,10,11,12}\rangle$ & $c_{222,002},c_{222,001},c_{222,022},c_{222,021}$ \\
     \hline
     28 & $|\psi_{15,16}\rangle,|\phi_{17,18,19,20}\rangle$ & $c_{102,101},c_{102,111},c_{102,201},c_{102,211}$ & 77 & $|\varphi_{8}\rangle,|\psi_{1,2}\rangle$ & $c_{111,101},c_{111,102}$\\
     \hline
     29 & $|\psi_{19}\rangle,|\phi_{29,30}\rangle$ & $c_{110,112},c_{110,012}$ & 78 & $|\varphi_{8}\rangle,|\phi_{5,6,7,8}\rangle$ & $c_{111,211},c_{111,210},c_{111,221},c_{111,220}$\\
     \hline
     30 & $|\psi_{15,16}\rangle,|\varphi_2\rangle$ & $c_{121,122},c_{121,102}$ & 79 & $|\varphi_{3}\rangle,|\psi_{17,18}\rangle$ & $c_{202,200},c_{202,210}$\\
     \hline
     31 & $|\psi_{7,8}\rangle,|\psi_{17,18}\rangle$ & $c_{200,111},c_{200,211},c_{210,211},c_{210,111}$ & 80 & $|\psi_{17,18}\rangle,|\phi_{21,22,23,24}\rangle$ & $c_{210,212},c_{210,222},c_{210,012},c_{210,022}$\\
     \hline
     32 & $|\varphi_{9}\rangle,|\psi_{3,4}\rangle$ & $c_{222,212},c_{222,210}$ & 81 & $|\psi_{21}\rangle,|\phi_{33,34}\rangle$ & $c_{221,220},c_{221,120}$\\
     \hline
    \end{tabular}
    \caption{Off-Diagonal elements of $\Pi_{BCD}$}
    \label{tab:3333_irreducible3}
\end{table*}

Furthermore, we will show that all diagonal elements in $\Pi_{BCD}$ are equal in Table \ref{tab:3333_irreducible33}:
\begin{table}
  \centering
  \begin{tabular}{|c|c|c|}
     \hline
     No. & States & Elements \\
     \hline
     \hline
     1 & $|\psi_{1,2}\rangle$ & $c_{101,101}=c_{102,102}$ \\
     \hline
     2 & $|\psi_{3,4}\rangle$ & $c_{212,212}=c_{210,210}$ \\
     \hline
     3 & $|\psi_{5,6}\rangle$ & $c_{020,020}=c_{021,021}$ \\
     \hline
     4 & $|\psi_{7,8}\rangle$ &  $c_{111,111}=c_{211,211}$\\
     \hline
     5 & $|\psi_{9,10}\rangle$ &  $c_{222,222}=c_{022,022}$ \\
     \hline
     6 & $|\psi_{11,12}\rangle$ &  $c_{000,000}=c_{100,100}$\\
     \hline
     7 & $|\psi_{13,14}\rangle$ &  $c_{011,011}=c_{021,021}$\\
     \hline
     8 & $|\psi_{15,16}\rangle$ &  $c_{122,122}=c_{102,102}$\\
     \hline
     9 & $|\psi_{17,18}\rangle$ &  $c_{200,200}=c_{210,210}$\\
     \hline
     10 & $|\phi_{1,2,3,4}\rangle$ & $c_{100,100}=c_{102,102}=c_{110,110}=c_{112,112}$\\
     \hline
     11 & $|\phi_{5,6,7,8}\rangle$ & $c_{211,211}=c_{21,210}=c_{221,221}=c_{220,220}$\\
     \hline
     12 & $|\phi_{9,10,11,12}\rangle$ & $c_{002,002}=c_{001,001}=c_{022,022}=c_{021,021}$\\
     \hline
     13 & $|\phi_{13,14,15,16}\rangle$ & $c_{000,000}=c_{020,020}=c_{100,100}=c_{120,120}$\\
     \hline
     14 & $|\phi_{17,18,19,20}\rangle$ & $c_{111,111}=c_{101,101}=c_{211,211}=c_{201,201}$\\
     \hline
     15 & $|\phi_{21,22,23,24}\rangle$ & $c_{222,222}=c_{212,212}=c_{022,022}=c_{012,012}$\\
     \hline
     16 & $|\phi_{25,26}\rangle$ & $c_{001,001}=c_{201,201}$\\
     \hline
     17 & $|\phi_{29,30}\rangle$ & $c_{112,112}=c_{012,012}$\\
     \hline
     18 & $|\phi_{33,34}\rangle$ & $c_{220,220}=c_{120,120}$\\
     \hline
     19 & $|\phi_{37,38}\rangle$ & $c_{010,010}=c_{011,011}$\\
     \hline
     20 & $|\phi_{41,42}\rangle$ & $c_{121,121}=c_{122,122}$\\
     \hline
     21 & $|\phi_{45,26}\rangle$ & $c_{202,202}=c_{200,200}$\\
     \hline
   \end{tabular}
  \caption{Diagonal elements of $\Pi_{BCD}$}\label{tab:3333_irreducible33}
\end{table}

Since $\Pi_{BCD}$ is proportional to an identity operator, BCD can not go first. On account that (\ref{3333_nonlocal}) is invariant under the cyclic permutation of the parties, $ACD$, $ABD$ can not go fist as well.

\end{document}